\documentclass[12pt,draftcls,onecolumn]{IEEEtran}

\usepackage[utf8]{inputenc} 
\usepackage[T1]{fontenc}
\usepackage{url}
\usepackage{ifthen}
\usepackage{cite}
\usepackage{mathrsfs}
\usepackage{color}
\usepackage[cmex10]{amsmath} 

\usepackage{amssymb,amsthm,acronym,graphicx,dsfont,stmaryrd,paralist}

\newcommand{\E}[2][]{{\mathbb{E}_{#1}}{\left(#2\right)}}       
\renewcommand{\P}[2][]{{\mathbb{P}_{#1}}{\left(#2\right)}}
       
\newcommand{\avgD}[2]{{{\mathbb{D}}\!\left({#1\Vert#2}\right)}}

\newcommand{\abs}[1]{\ensuremath{\left|#1\right|}}              
\newcommand{\eqdef}{\ensuremath{\triangleq}}                    

\renewcommand{\leq}{\leqslant}
\renewcommand{\geq}{\geqslant}

\DeclareMathAlphabet{\eurm}{U}{eur}{m}{n}
\DeclareMathAlphabet{\mathbsf}{OT1}{cmss}{bx}{n}
\DeclareMathAlphabet{\mathssf}{OT1}{cmss}{m}{sl}
\DeclareMathAlphabet{\mathcsf}{OT1}{cmss}{sbc}{n}

\DeclareSymbolFont{bsfletters}{OT1}{cmss}{bx}{n}  
\DeclareSymbolFont{ssfletters}{OT1}{cmss}{m}{n}
\DeclareMathSymbol{\bsfGamma}{0}{bsfletters}{'000}
\DeclareMathSymbol{\ssfGamma}{0}{ssfletters}{'000}
\DeclareMathSymbol{\bsfDelta}{0}{bsfletters}{'001}
\DeclareMathSymbol{\ssfDelta}{0}{ssfletters}{'001}
\DeclareMathSymbol{\bsfTheta}{0}{bsfletters}{'002}
\DeclareMathSymbol{\ssfTheta}{0}{ssfletters}{'002}
\DeclareMathSymbol{\bsfLambda}{0}{bsfletters}{'003}
\DeclareMathSymbol{\ssfLambda}{0}{ssfletters}{'003}
\DeclareMathSymbol{\bsfXi}{0}{bsfletters}{'004}
\DeclareMathSymbol{\ssfXi}{0}{ssfletters}{'004}
\DeclareMathSymbol{\bsfPi}{0}{bsfletters}{'005}
\DeclareMathSymbol{\ssfPi}{0}{ssfletters}{'005}
\DeclareMathSymbol{\bsfSigma}{0}{bsfletters}{'006}
\DeclareMathSymbol{\ssfSigma}{0}{ssfletters}{'006}
\DeclareMathSymbol{\bsfUpsilon}{0}{bsfletters}{'007}
\DeclareMathSymbol{\ssfUpsilon}{0}{ssfletters}{'007}
\DeclareMathSymbol{\bsfPhi}{0}{bsfletters}{'010}
\DeclareMathSymbol{\ssfPhi}{0}{ssfletters}{'010}
\DeclareMathSymbol{\bsfPsi}{0}{bsfletters}{'011}
\DeclareMathSymbol{\ssfPsi}{0}{ssfletters}{'011}
\DeclareMathSymbol{\bsfOmega}{0}{bsfletters}{'012}
\DeclareMathSymbol{\ssfOmega}{0}{ssfletters}{'012}

\newcommand{\calM}{{\mathcal{M}}}

\newcommand{\calQ}{{\mathcal{Q}}}

\newcommand{\calX}{{\mathcal{X}}}

\acrodef{AEP}{Asymptotic Equipartition Property}
\acrodef{AoA}{Angle of Arrival}
\acrodef{AWGN}{Additive White Gaussian Noise}
\acrodef{BER}{Bit-Error-Rate}
\acrodef{BEC}{Binary Erasure Channel}
\acrodef{BPSK}{Binary Phase-Shift Keying}
\acrodef{BSC}{Binary Symmetric Channel}
\acrodef{CDF}[CDF]{Cumulative Distribution Function}
\acrodef{CLT}[CLT]{Central Limit Theorem}
\acrodef{CSI}[CSI]{Channel State Information}
\acrodef{DMC}[DMC]{Discrete Memoryless Channel}
\acrodef{DMS}[DMS]{Discrete Memoryless Source}
\acrodef{iid}[i.i.d.]{independent and identically distributed}
\acrodef{LPD}[LPD]{Low Probability of Detection}
\acrodef{LDPC}[LDPC]{Low-Density Parity-Check}
\acrodef{MAC}[MAC]{multiple-access channel}
\acrodef{MIMO}[MIMO]{Multiple-Input Multiple-Output}
\acrodef{MISO}{Multiple-Input Single-Output}
\acrodef{PDF}[PDF]{Probability Distribution Function}
\acrodef{PMF}[PMF]{Probability Mass Function}
\acrodef{PPM}[PPM]{Pulse Position Modulation}
\acrodef{PSD}{Power Spectral Density}
\acrodef{QPSK}{Quadrature Phase-Shift Keying}
\acrodef{SIMO}{Single-Input Multiple-Output}
\acrodef{SNR}{Signal-to-Noise Ratio}
\acrodef{wrt}[w.r.t.]{with respect to}
\acrodef{WSS}{Wide Sense Stationary}

\graphicspath{{graphics/}}

\newtheorem{theorem}{Theorem}
\newtheorem{remark}{\textbf{Remark}}
\newtheorem{definition}{\textbf{Definition}}
\newtheorem{lemma}{Lemma}
\newtheorem{proposition}{Proposition}

\makeatletter
\if@twocolumn%
\newcommand{\hspaceonetwocol}[2]{\hspace{#2}}
\newcommand{\includeonetwocol}[2]{#2}
\def\twocolbreak{\nonumber\\ & }%
\def\linesplit{\nonumber\\ }%
\else
\newcommand{\hspaceonetwocol}[2]{\hspace{#1}}
\newcommand{\includeonetwocol}[2]{#1}
\def\twocolbreak{}%
\def\linesplit{}%
\fi%
\makeatother%

\allowdisplaybreaks

\acrodef{KL}[KL]{Kullback-Leibler}

\begin{document}
\sloppy

\title{Cooperative Resolvability and Secrecy in the Cribbing Multiple-Access Channel}
\author{Noha Helal, {\em Student Member,~IEEE}, Matthieu Bloch, {\em Senior Member,~IEEE},\\ and Aria Nosratinia, {\em Fellow,~IEEE}
\thanks{Noha Helal and Aria Nosratinia are with the department of Electrical Engineering, University of Texas at Dallas, Dallas, TX 75080 USA, (e-mail: noha.helal@utdallas.edu and aria@utdallas.edu). Matthieu Bloch is with the School of Electrical and Computer Engineering, Georgia Institute of Technology, Atlanta, GA 30332 USA, (e-mail: matthieu.bloch@ece.gatech.edu).}
\thanks{This work was presented in part in ISIT 2018.}
}

\maketitle

\begin{abstract}

 We study channel resolvability for the discrete memoryless multiple-access channel with cribbing, i.e., the characterization of the amount of randomness required at the  inputs to approximately produce a chosen i.i.d. output distribution according to \acl{KL} divergence. We analyze resolvability rates when one encoder cribs (i) the input of the other encoder; or the output of the other encoder, (ii)  non-causally, (iii) causally, or (iv) strictly-causally.
For scenarios (i)-(iii), we exactly characterize the channel resolvability region. For (iv), we provide inner and outer bounds for the channel resolvability region; the crux of our achievability result is to  handle the strict causality constraint with a block-Markov coding scheme in which dependencies across blocks are suitably hidden. Finally, we leverage the channel resolvability results to derive achievable secrecy rate regions for each of the cribbing scenarios under strong secrecy constraints.

\end{abstract}

\section{Introduction}
\label{intro}

Producing an approximation of a desired statistic at a channel output via application of minimal randomness at its input has emerged as a useful building block for many problems in information theory, including source coding~\cite{resolvability_coding}, rate distortion theory~\cite{resolvability_distortion}, coordination~\cite{synthesis} and strong secrecy~\cite{Hayashi2006,strong_broadcast_resolvability,strong_interference_resolvability_journal,strong_secrecy_cooperative,e_resolvability}. The origins of this problem can be traced back to Wyner's work~\cite{wyner_common} on the characterization of common randomness among two dependent random variables, in which he used a normalized \ac{KL} divergence as a measure of approximation. The problem was subsequently formalized and generalized for total variation using information-spectrum methods~\cite{resolvability} and the optimal amount of randomness was called {\em channel resolvability.} Subsequent works have simplified proofs, both for total variation~\cite{synthesis} and \ac{KL} divergence~\cite{Hou2013}, and studied multi-user settings, such as \acfp{MAC} with non-cooperative encoders~\cite{resolvability_mac,Frey2017,strongmac,strong_twoway}.

This paper studies the channel resolvability region for discrete memoryless \acfp{MAC} with cooperating encoders; specifically, we focus on a cribbing model~\cite{cribbingmeulen} in which one encoder has access to the input or output of the other encoder subject to various causality constraints. The cribbing model captures the \emph{essence} of cooperation and produces results and insights that are independent of cooperation signaling mechanisms. 
A key contribution of this paper is the characterization of the channel resolvability region for four cribbing situations while approximating the output distribution according to non-normalized \ac{KL} divergence:
(i)~ degraded message sets,
(ii) non-causal cribbing,
(iii) causal cribbing and 
(iv) strictly-causal cribbing.

In the case of \emph{non-causal} cribbing, an achievability region is already available from~\cite[Corollary VII.8]{synthesis} for approximation in terms of total variation, in which case our contribution is only to derive an achievability and matching converse for \ac{KL} divergence. In the case of \emph{causal} and \emph{strictly-causal cribbing}, we not only develop converse results but also propose an achievability scheme that handles the causality constraints with a block-Markov coding scheme; in particular, we show that the dependencies created by the block-Markov coding scheme can be \emph{hidden} through appropriate recycling of secret randomness after identifying a wiretap channel embedded in the model.

In the second part of this paper, strong secrecy~\cite{almost_indep,weak2strong} achievable rates are obtained for the \ac{MAC} with cribbing.  We develop a wiretap coding scheme that achieves strong secrecy, fueled by the results in the earlier part of this paper on the resolvability of cribbing \ac{MAC}. The secrecy rates for MAC with degraded message sets as well as non-causal cribbing follow from corresponding resolvability results without complication. However, a novel approach is needed for the causal and strictly causal cribbing, because the cribbing MAC codebook devised by Willems and van der Muelen~\cite{cribbingmeulen} is not directly compatible with the randomness recycling in the resolvability codebooks under causality constraints. Accordingly, a notable contribution of this work is a new superposition coding strategy that uses all components of the cribbing signal for cooperation to achieve efficient decoding at the legitimate receiver,  while at the same time forcing a part of the randomness within the cribbing signal to remain non-cooperative {\em from the viewpoint of the eavesdropper.} This feature is needed for randomness recycling and is crucial to avoid secrecy rate loss under strictly causal and causal cribbing.

Strong secrecy for \ac{MAC} with {\em non-cooperating} encoders has been studied in~\cite{Frey2017,strongmac,strong_twoway}. To the best of our knowledge, strong secrecy in multi-terminal settings under cooperating encoders has not been comprehensively studied. For completeness we highlight examples of the investigation of  {\em weak} secrecy under cooperation: \ac{MAC} with cooperating or partially cooperating encoders~\cite{simeonecognitive,mac_conference,tang2007multiple,Awan:TIFS13}, interference channel with cooperating encoders~\cite{interference_gf,Liang:IT09,Farsani:isit2014}, relay channel with an external eavesdropper~\cite{relaywt,relay_wt} and broadcast channel with cooperating receivers~\cite{coop_relay_bc}. These works follow the classical approach of Wyner~\cite{wyner_wiretap} and Csisz\'ar~\cite{csiszar_wiretap} to develop weak secrecy results. To the best of our understanding there has been limited work on any type of cooperative strong secrecy; notable examples are Goldfeld {\em et al.}~\cite{strong_secrecy_cooperative} on {\em receive side} cooperation in the broadcast channel, Watanabe and Oohama~\cite{Watanabe:IT14} on the cognitive interference channel with confidential messages, and Chou and Yener~\cite{Chou:ISIT16} on Polar coding for the MAC wiretap channel with cooperative jamming.

The remainder of this paper is organized as follows. Section~\ref{sec:notation} clarifies our notation, Section~\ref{sec:def-main} introduces the problem definition and presents the main resolvability results of the paper, Section~\ref{sec:proofs} produces the resolvability proofs, and Section~\ref{sec:strong_sec} presents the strong secrecy results. 

\begin{figure*}
\begin{minipage}{0.5\textwidth}
  \begin{center}
    \includegraphics[width=0.9\textwidth]{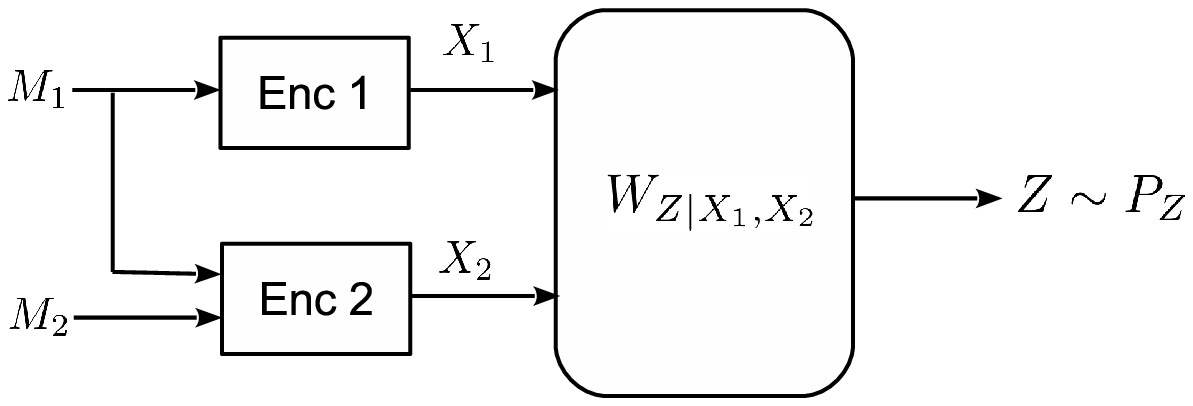} \\[0.1in]
{\footnotesize (a) Degraded message sets \ac{MAC}}
  \end{center}
\end{minipage}
\begin{minipage}{0.5\textwidth}
  \begin{center}
    \includegraphics[width=0.9\textwidth]{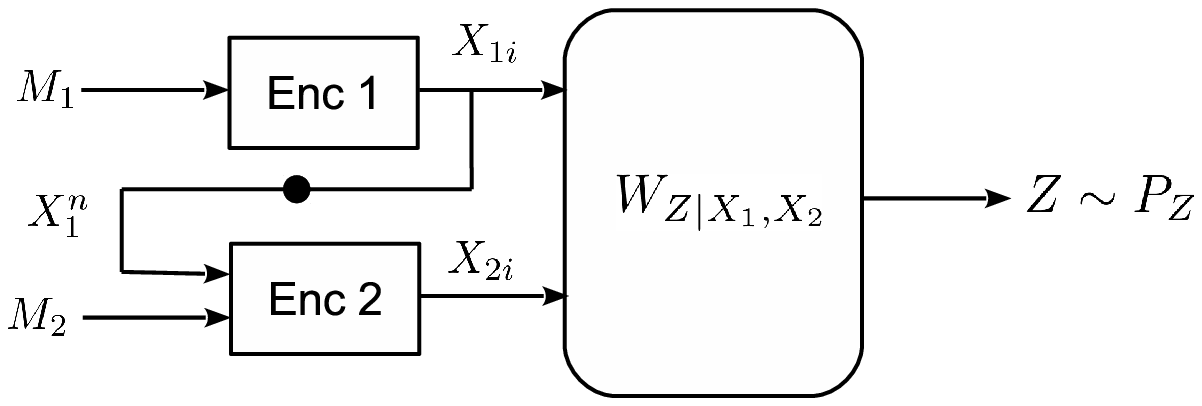} \\[0.1in]
{\footnotesize (b) Non-causal cribbing \ac{MAC} }
  \end{center}
\end{minipage}\\[0.3in]
\begin{minipage}{0.5\textwidth}
  \begin{center}
    \includegraphics[width=0.9\textwidth]{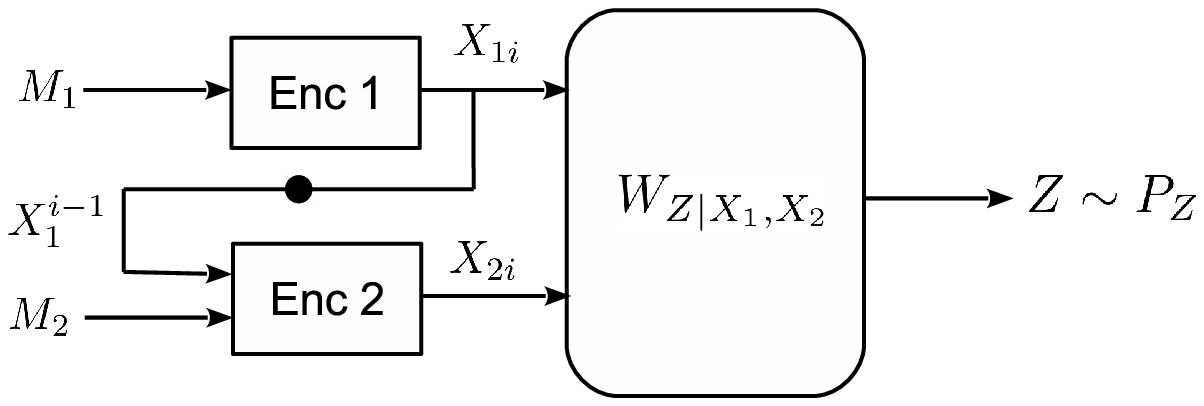} \\[0.1in]
{\footnotesize (c) Strictly-causal cribbing \ac{MAC}}
  \end{center}
\end{minipage}
\begin{minipage}{0.5\textwidth}
  \begin{center}
    \includegraphics[width=0.9\textwidth]{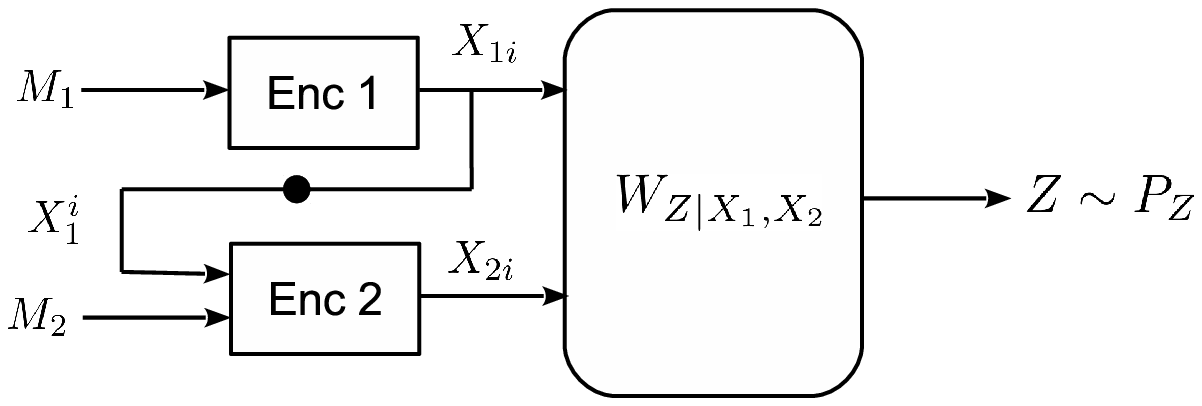} \\[0.1in]
{\footnotesize (d) Causal cribbing \ac{MAC}}
  \end{center}
\end{minipage}\\[0.2in]
\caption{Cribbing \ac{MAC} scenarios}
\label{fig1}
\end{figure*}

\section{Notation}
\label{sec:notation}

Random variables are represented by upper case letters and their realizations by the corresponding lower case letters, e.g., $x$ is a realization of the random variable $X$. Superscripts denote the length of a sequence of symbols and subscripts denote the position of a symbol in a sequence. Calligraphic letters represent sets, the cardinality of which is denoted by $|\cdot|$. For example, $X^n=\{X_1,\dots,X_n\}$ where $X_i$ belongs to the alphabet $\mathcal{X}$ of size $|\mathcal{X}|$.
$P_X$ and $P_{XY}$ denote probability distributions on $\mathcal{X}$ and $\mathcal{X}\times\mathcal{Y}$, respectively. We sometimes omit the subscripts in probability distributions if they are clear from the context, i.e., we write $P(x)$ instead of $P_X(x)$.

For two distributions $P$ and $Q$ on the same alphabet, the \ac{KL} divergence is defined by $\avgD{P}{Q} \triangleq \sum_x P(x) \log \frac{P(x)}{Q(x)}$ and the variational distance is defined by $\mathbb{V}(P,Q)\triangleq\sum_x |P(x)-Q(x)|$. Throughout the paper, $\log$ denotes the base~2 logarithm. For an i.i.d.\ vector whose components are distributed according to $P_X$, the product distribution is denoted by $P_X^{\otimes n}(x^n)\eqdef\prod_{i=1}^n P_X(x_i)$. The set of $\epsilon$-strongly-typical sequences of length $n$ with respect to $P_X$ is defined as:
\[
\mathcal{T}^n_\epsilon(P_X) \triangleq \Big\{x^n: \abs{\frac{N(a|x^n)}{n}-P_X(a)} \leq \epsilon P_X(a), \forall a \in \mathcal{X}  \Big\}.
\]
For a set of random variables $\{X_i\}_{i\in\calM}$ indexed over a countable set $\calM$, $\E[\setminus m]{\cdot}$ denotes the expectation over all random variables with indices in $\mathcal{M}$ except that with index $m$. $\E[M]{\cdot}$ is the expectation w.r.t. the random variable $M$ and  $\mathds{1}_{\{\cdot\}}$ is the indicator function.

\section{Problem Definition and Main Results}
\label{sec:def-main}

\subsection{Problem Definition: Channel Resolvability of \ac{MAC} with Cribbing}
\label{sec:mac-resolv-with}

We consider four scenarios for the two-user discrete memoryless \ac{MAC} with cribbing (see~Fig.~\ref{fig1}). The discrete memoryless \ac{MAC} $(\mathcal{X}_1 \times \mathcal{X}_2, W_{Z|X_1,X_2},\mathcal{Z})$ consists of finite input alphabets~$\mathcal{X}_1$ and $\mathcal{X}_2$, and finite output alphabet~$\mathcal{Z}$, together with a channel transition probability $W_{Z|X_1X_2}$. 
For a joint input distribution $P_{X_1,X_2}$ on $\calX_1\times\calX_2$, the output is distributed according to $Q_Z(z)=\sum_{x_1,x_2} P_{X_1,X_2}(x_1,x_2) W_{Z|X_1,X_2}(z|x_1,x_2)$.
A $(2^{nR_1},2^{nR_2},n)$ channel resolvability code consists of two encoders $f_1$ and $f_2$ with inputs $M_1$ and $M_2$ defined on  $\mathcal{M}_1= \llbracket 1,2^{nR_1} \rrbracket$ and $\mathcal{M}_2=\llbracket 1,2^{nR_2} \rrbracket$. In the four scenarios studied in this paper, the per-symbol encoding functions are defined as follows.

In the \ac{MAC} with degraded message sets (Fig.~\ref{fig1}(a)), Encoder~2 has access to the input of Encoder~1 
\begin{align}
f_{1i}: \mathcal{M}_1 \rightarrow \mathcal{X}_1 \qquad
f_{2i}: \mathcal{M}_1 \times \mathcal{M}_2 \rightarrow \mathcal{X}_2. \label{encfn_degraded}
\end{align}
In the \ac{MAC} with non-causal cribbing  (Fig.~\ref{fig1}(b)), Encoder~2 has non-causal access to the entire current codeword of Encoder~1 
\begin{align}
f_{1i}: \mathcal{M}_1 \rightarrow \mathcal{X}_1    \qquad
f_{2i}: \mathcal{M}_2 \times \mathcal{X}_1^n \rightarrow \mathcal{X}_2. \label{encfn_noncausal}
\end{align}
In the \ac{MAC} with strictly-causal cribbing  (Fig.~\ref{fig1}(c)), Encoder~2 has access to the output of Encoder~1 with a one-symbol delay
\begin{align}
f_{1i}: \mathcal{M}_1 \rightarrow \mathcal{X}_1    \qquad 
f_{2i}: \mathcal{M}_2 \times \mathcal{X}^{i-1} \rightarrow \mathcal{X}_2. \label{encfn_strict}
\end{align}
In the \ac{MAC} with causal cribbing (Fig.~\ref{fig1}(d)), Encoder~2 has access to the output of Encoder~1 with zero delay
\begin{align}
f_{1i}: \mathcal{M}_1 \rightarrow \mathcal{X}_1     \qquad
f_{2i}: \mathcal{M}_2 \times \mathcal{X}_1^{i} \rightarrow \mathcal{X}_2. \label{encfn_causal}
\end{align}

\begin{definition}
A rate pair $(R_1,R_2)$ is said to be achievable for the discrete memoryless \ac{MAC} $(\mathcal{X}_1 \times \mathcal{X}_2, W_{Z|X_1X_2},\mathcal{Z})$ if for a given\footnote{Originally in~\cite{resolvability} the resolvability rate was defined as ``the number of random bits required per channel use in order to generate an input that achieves arbitrarily accurate approximation of the output statistics {\em for any given input process}.'' More recent works consider a fixed but arbitrary $Q_Z$, leaving out the implied intersection of rate regions over different $Q_Z$. This is more convenient in several ways, including in the application of resolvability to secrecy problems where only the simulation of a certain $Q_Z$ is of interest.}  $Q_Z$ there exists a sequence of $(2^{nR_1},2^{nR_2},n)$ codes  such that $\lim_{n\to\infty}\avgD{P_{Z^n}}{Q_Z^{\otimes n}}=0$. The \ac{MAC} resolvability region is the closure of the set of achievable rate pairs $(R_1,R_2)$.
\end{definition}

\subsection{Statement of Main Results}
\label{sec:main-results} 
We first recall the known \ac{MAC} resolvability region with non-cooperating encoders, which will serve as a reference to assess the benefits of cooperation.
\begin{theorem}[\cite{Frey2017}] 
\label{th_noncoop}
The resolvability region for the \ac{MAC} with non-cooperating encoders is the set of rate pairs $(R_1,R_2)$ such that 
\begin{align}
R_1 &\geq I(X_1;Z|V),\\
R_2 &\geq I(X_2;Z|V),\\
R_1+R_2 &\geq I(X_1,X_2;Z|V),
\end{align}
for some joint distribution $P_{VX_1X_2Z}\eqdef P_VP_{X_1|V}P_{X_2|V}W_{Z|X_1X_2}$ with marginal $Q_Z$.
\end{theorem}
The resolvability region in~\cite{Frey2017} was derived with respect to variational distance approximation. 

\begin{theorem} 
\label{th_degraded}
The resolvability region for the \ac{MAC}  with degraded message sets is the set of rate pairs $(R_1,R_2)$ such that 
\begin{align}
R_1 &\geq I(X_1;Z),\\
R_1+R_2&\geq I(X_1,X_2;Z),
\end{align}
for some joint distribution $P_{X_1X_2Z}\eqdef P_{X_1X_2}W_{Z|X_1X_2}$ with marginal $Q_Z$.
\end{theorem}

\begin{proof}
See Section~\ref{degmsg}.
\end{proof}
The absence of the individual rate constraint $R_2$ is a direct consequence of the degraded message model that allows Encoder~2 to benefit from the randomness provided via $R_1$, reducing the {\em required} individual rate constraint for User~2 to $R_2\ge 0$, which is omitted. Another difference between the regions described by Theorem~\ref{th_noncoop} and Theorem~\ref{th_degraded} is that the former includes a convexifying auxiliary random variable that is missing in the latter. The region described by Theorem~\ref{th_degraded} is already convex due to the larger set of available input distributions, as proved in Appendix~\ref{appendix:convexity_degraded}.

\begin{theorem}\label{th_noncausal}
The resolvability region for the \ac{MAC} with non-causal cribbing is the set of rate pairs $(R_1,R_2)$ such that 
\begin{align}
R_1 &\geq I(X_1;Z),\\
R_2&\geq I(X_1,X_2;Z)-H(X_1),\\
R_1+R_2&\geq I(X_1,X_2;Z),
\end{align}
for some joint distribution $P_{X_1X_2Z}\eqdef P_{X_1X_2}W_{Z|X_1X_2}$ with marginal $Q_Z$.
\end{theorem}
\begin{proof}
See Section~\ref{noncausal}.
\end{proof}
The achievability result was derived in~\cite[Corollary VII.8]{synthesis} for approximation of the output statistics in terms of variational distance. Our contribution here is to provide achievability and converse proofs for approximation in terms of \ac{KL} divergence. Compared with the \ac{MAC} with degraded message sets, more randomness is required as seen by the presence of an individual rate constraint for User~2. As already discussed in~\cite{synthesis}, this stems from the impossibility for Encoder~2 to extract uniform randomness from the observation of the output of Encoder~1 at a rate exceeding $H(X_1)$. This rate region is convex (see Appendix~\ref{appendix:convexity_causal}).

\begin{theorem}
\label{th_strict}
The resolvability region for the  \ac{MAC} with  strictly-causal cribbing is included in the set of rate pairs $(R_1, R_2)$ such that
\begin{align}
R_1 &\geq I(U,X_1;Z),\\
R_2&\geq I(X_1,X_2;Z)-H(X_1|U),\\
R_1+R_2&\geq I(X_1,X_2;Z),
\end{align}
for some joint distribution $P_{UX_1X_2Z}\eqdef P_{U}P_{X_1|U}P_{X_2|U}W_{Z|X_1X_2}$ with marginal $Q_Z$. An achievable
region is characterized by the same rate constraints and distribution, but subject to the additional
constraint $H(X_1|U) > I(U,X_1;Z)$.
\end{theorem}

\begin{proof}
See Section~\ref{strictlycausal}.
\end{proof}
The strict causality constraint leads to the appearance of a new random variable $U$ in several terms in the right hand side of the individual rate constraints, thus the individual rate constraints can be larger than Theorem~\ref{th_noncausal}, i.e., more needed randomness. The achievable region provided by Theorem~\ref{th_strict} does not provably match the outer bound because the set of probability distributions available in the achievable region is smaller than in the converse, due to the additional constraint $H(X_1|U) > I(U,X_1;Z)$. The rate regions defined by the inner and outer bounds are convex (see Appendix~\ref{appendix:convexity_strict}).

\begin{theorem}\label{th_causal}
The resolvability region for the  \ac{MAC} with  causal cribbing is the set of rate pairs $(R_1,R_2)$ such that
\begin{align}
R_1 &\geq I(X_1;Z)\\
R_2&\geq I(X_1,X_2;Z)-H(X_1)\\
R_1+R_2&\geq I(X_1,X_2;Z)
\end{align}
for some joint distribution $P_{X_1X_2Z}\eqdef P_{X_1X_2}W_{Z|X_1X_2}$ with marginal $Q_Z$.
\end{theorem}

\begin{proof}
See Section~\ref{causal}.
\end{proof}
Similar to what was observed for the \ac{MAC} reliability region~\cite{cribbingmeulen}, the \ac{MAC} resolvability region with causal cribbing is \emph{identical} to that obtained for non-causal cribbing. 
\begin{remark}
  Theorem~\ref{th_causal} is established from Theorem~\ref{th_strict} using Shannon strategies as in~\cite{cribbingmeulen}, yet Theorem~\ref{th_causal} has an achievable region that is tight against the outer bound when Theorem~\ref{th_strict} does not. Perhaps surprisingly, this happens because the choice of random variables in the Shannon strategy automatically satisfies the constraint imposed in the achievability of Theorem~\ref{th_strict}.
\end{remark}


\section{Strong Secrecy from Channel Resolvability}
\label{sec:strong_sec}

In this section we use the resolvability results to study the multiple-access wiretap channel with cribbing. For each of the cribbing models previously discussed, an achievable strong secrecy rate region is presented.
Consider a \ac{MAC} with cribbing $(\mathcal{X}_1 \times \mathcal{X}_2, W_{Z|X_1,X_2},\mathcal{Y},\mathcal{Z})$ where $\mathcal{X}_1$ and $\mathcal{X}_2$ are finite input alphabets, $\mathcal{Y}$ and $\mathcal{Z}$ are the finite output alphabets of the legitimate receiver and the wiretapper, respectively. A $(2^{nR_1},2^{nR_2},n)$ code consists of two encoders $f_1$ and $f_2$ and a decoder $g$. The encoders are defined similar to~\eqref{encfn_degraded}-\eqref{encfn_causal} but the functions $f_{1i}$ and $f_{2i}$ are now stochastic and not deterministic. The decoding function at the legitimate receiver is defined as:
\begin{align}
g:\mathcal{Y}^n  \rightarrow  \mathcal{\hat{M}}_1 \times \mathcal{\hat{M}}_2.
\end{align}

 The probability of error at the legitimate receiver is defined as $P_e^{(n)}=\mathbb{P}\Big( (\hat{M}_1,\hat{M}_2) \neq (M_1,M_2) \Big)$.
The strong secrecy metric adopted in this paper is the total amount of leaked confidential information per codeword, defined as $L^{(n)}=I(M_1,M_2;Z^n)$.

\begin{definition}
A strong secrecy rate pair $(R_1,R_2)$ is said to be achievable for the discrete memoryless \ac{MAC} $(\mathcal{X}_1 \times \mathcal{X}_2, W_{Z|X_1X_2},\mathcal{Y},\mathcal{Z})$ if there exists a sequence of $(2^{nR_1},2^{nR_2},n)$ codes such that $P_e^{(n)}$ and $L^{(n)}$ vanish as $n\to\infty$.
\end{definition}

\begin{proposition} 
\label{degraded_secrecy}
 For the multiple-access wiretap channel with degraded message sets, the following strong-secrecy rate region is achievable: 
\[
(R_1,R_2) = \bigcup_{ P_{X_1X_2} W_{{YZ}|X_1X_2}} \mathcal{R}_{\textnormal{DM}}^{(\textnormal{in})},\\
\]
\begin{equation}
\mathcal{R}_{\textnormal{DM}}^{(\textnormal{in})}=
\left\{ 
\begin{array}{ll}
R_1, R_2 \geq 0\\
R_2 \leq  I(X_2;Y|X_1) \\
R_1+R_2 \leq I(X_1,X_2;Y)-I(X_1,X_2;Z)\\
\end{array} 
\right\}. 
\label{eq:proposition1}
\end{equation}
\end{proposition}

\begin{proof}
See Appendix~\ref{appendix:mac_degraded}.
\end{proof}


\begin{proposition} \label{noncausal_secrecy}
 For the multiple-access wiretap channel with non-causal cribbing, the following strong-secrecy rate region is achievable: 
\[
(R_1,R_2) = \bigcup_{ P_{X_1X_2} W_{{YZ}|X_1X_2}}\mathcal{R}_{\textnormal{NC}}^{(\textnormal{in})},\\
\]
\begin{equation}
\mathcal{R}_{\textnormal{NC}}^{(\textnormal{in})}=
\left\{ 
\begin{array}{ll}
R_1, R_2 \geq 0\\
R_1 \leq H(X_1) -I(X_1;Z) \\
R_2 \leq  I(X_2;Y|X_1) \\
R_1+R_2 \leq I(X_1,X_2;Y)-I(X_1,X_2;Z)\\
\end{array} 
\right\}.
\label{eq:proposition2}
\end{equation}
\end{proposition}

\begin{proof}
See Appendix~\ref{appendix:mac_noncausal}.

\end{proof}


\begin{proposition} 
\label{strictly_secrecy}
 For the multiple-access wiretap channel with strictly-causal cribbing, the following strong-secrecy rate region is achievable: 
\[
(R_1,R_2) = \bigcup_{ P_U P_{X_1|U} P_{X_2|U} W_{{YZ}|X_1X_2}} \mathcal{R}_{\textnormal{SC}}^{(\textnormal{in})},\\
\]
\begin{equation}
\mathcal{R}_{\textnormal{SC}}^{(\textnormal{in})}=
\left\{ 
\begin{array}{ll}
R_1, R_2 \geq 0\\
R_1 \leq  H(X_1|U) - I(U,X_1;Z)  \\
R_2 \leq  I(X_2;Y|X_1,U) \\
R_1+R_2 \leq H(X_1|U)+I(X_2;Y|X_1,U) -I(X_1,X_2;Z)\\
R_1+R_2 \leq I(X_1,X_2;Y)- I(X_1,X_2;Z) 
\end{array} 
\right\}.
\label{eq:proposition3}
\end{equation}
\end{proposition}

\begin{proof}

See Appendix~\ref{appendix:mac_strictlycausal}.

\end{proof}

\begin{remark}
Recall the resolvability achievable rate region under the strictly causal cribbing had a constraint $H(X_1|U)> I(U,X_1;Z)$ on the allowable probability distributions. This constraint is implicit in Proposition~\ref{strictly_secrecy} in the form of the non-negativity constraint on $R_1$.
\end{remark}


\begin{proposition}  \label{causal_secrecy}
 For the multiple-access wiretap channel with causal cribbing, the following strong-secrecy rate region is achievable: 
\[
(R_1,R_2) = \bigcup_{ P_{X_1X_2} W_{{YZ}|X_1X_2}} \mathcal{R}_{\textnormal{C}}^{(\textnormal{in})},\\
\]
\begin{equation}
\mathcal{R}_{\textnormal{C}}^{(\textnormal{in})}=
\left\{ 
\begin{array}{ll}
R_1, R_2 \geq 0\\
R_1 \leq  H(X_1) - I(X_1;Z)  \\
R_2 \leq  I(X_2;Y|X_1) \\
R_1+R_2 \leq I(X_1,X_2;Y)- I(X_1,X_2;Z) 
\end{array} 
\right\}.
\label{eq:proposition4}
\end{equation}
\end{proposition}

\begin{proof}

See Appendix~\ref{appendix:mac_causal}.

\end{proof}


\section{Proofs of Main Results}
\label{sec:proofs}

\subsection{Theorem~\ref{th_degraded}: \ac{MAC} with Degraded Message Sets}
\label{degmsg}

\subsubsection{Achievability} \hfill\\
Consider a distribution $P(x_1,x_2)=P(x_1)P(x_2|x_1)$ such that $\sum_{x_1,x_2} P(x_1,x_2)W(z|x_1,x_2)=Q_Z(z)$.
\begin{itemize}
\item Independently generate $2^{nR_1}$ codewords $x_1^n$ each with probability $P(x_1^n)=P_{X_1}^{\otimes n}(x_1^n)$. Label them $x_1^n(m_1)$, $m_1 \in \llbracket 1,2^{nR_1}\rrbracket$.
\item For every $m_1$, independently generate $2^{nR_2}$ codewords $x_2^n$ each with probability $P(x_2^n|x_1^n(m_1))=P_{X_2|X_1}^{\otimes n}(x_2^n|x_1^n(m_1))$. Label them $x_{2}^n(m_1,m_2)$,  $m_{2} \in \llbracket 1,2^{nR_2}\rrbracket$.
\end{itemize}

For $m_1\in \llbracket 1,2^{nR_1}\rrbracket$ and $m_{2} \in \llbracket 1,2^{nR_2}\rrbracket$,
let $X^n_1(m_1)$ and $X^n_2(m_1,m_2)$ denote the random variables representing the randomly generated codewords. The average \ac{KL} divergence is:

\begin{align}
&\mathbb{E} \big( \mathbb{D} ( P_{Z^n}||Q_Z^{\otimes n})\big) \nonumber \\
& = \mathbb{E}\left(\sum_{z^n} P_{Z^n}(z^n) \log \frac{P_{Z^n}(z^n)}{Q_Z^{\otimes n}(z^n)}\right)   \nonumber\\
& = \mathbb{E}\bigg(\!\! \sum_{z^n} \frac{1}{2^{n(R_1+R_2)}} \sum_{m_1} \sum_{m_2} W^{\otimes n}(z^n|X^n_1(m_1),X^n_2(m_1,m_2)) \nonumber\\ 
 &\hspaceonetwocol{3in}{0.5in} \log \frac{\sum_{m'_1} \sum_{m'_2} W^{\otimes n}(z^n|X^n_1(m'_1),X^n_2(m'_1,m'_2))}{2^{n(R_1+R_2)} Q_Z^{\otimes n}(z^n)}\bigg)  \nonumber\\
&\stackrel{(a)}{\leq} \frac{1}{2^{n(R_1+R_2)}} \sum_{m_1} \sum_{m_2} \sum_{z^n} \sum_{x_1^n(m_1)} \sum_{x_2^n(m_1,m_2)} \twocolbreak
 \hspaceonetwocol{0.0in}{0.2in} P(x_1^n(m_1),x_2^n(m_1,m_2),z^n) \nonumber\\
 &\hspaceonetwocol{2in}{0.2in}  \log \mathbb{E}_{\setminus (m_1,m_2)} \frac{\sum_{m'_1} \sum_{m'_2} W^{\otimes n}(z^n|X^n_1(m'_1),X^n_2(m'_1,m'_2))}{2^{n(R_1+R_2)} Q_Z^{\otimes n}(z^n)}  \nonumber\\
&= \frac{1}{2^{n(R_1+R_2)}} \sum_{m_1} \sum_{m_2} \sum_{z^n} \sum_{x_1^n(m_1)} \sum_{x_2^n(m_1,m_2)} \twocolbreak
 \hspaceonetwocol{0.0in}{0.2in}  P(x_1^n(m_1),x_2^n(m_1,m_2),z^n)  \nonumber\\
 &\hspaceonetwocol{0.2in}{0.2in}  \log \mathbb{E}_{\setminus (m_1,m_2)} \bigg( \frac{W^{\otimes n}(z^n|x^n_1(m_1),x^n_2(m_1,m_2))}{2^{n(R_1+R_2)} Q_Z^{\otimes n}(z^n)}  \twocolbreak
 \quad +\sum_{m'_2 \neq m_2} \frac{ W^{\otimes n}(z^n|x^n_1(m_1),X^n_2(m_1,m'_2))}{2^{n(R_1+R_2)} Q_Z^{\otimes n}(z^n)} \nonumber\\
&\hspaceonetwocol{3in}{0.2in}  +\sum_{m'_1 \neq m_1}\sum_{m'_2} \frac{ W^{\otimes n}(z^n|X^n_1(m'_1),X^n_2(m'_1,m'_2))}{2^{n(R_1+R_2)} Q_Z^{\otimes n}(z^n)}\bigg) \nonumber\\
&\leq  \frac{1}{2^{n(R_1+R_2)}} \sum_{m_1} \sum_{m_2} \sum_{z^n} \sum_{x_1^n(m_1)} \sum_{x_2^n(m_1,m_2)} \twocolbreak
 \hspaceonetwocol{0.0in}{0.2in} P^{\otimes n}(x_1^n(m_1),x_2^n(m_1,m_2),z^n)       \nonumber\\
&\hspaceonetwocol{0.2in}{0.2in}  \log \bigg( \frac{W^{\otimes n}(z^n|x^n_1(m_1),x^n_2(m_1,m_2))}{2^{n(R_1+R_2)} Q_Z^{\otimes n}(z^n)} \twocolbreak
 \hspaceonetwocol{0.0in}{0.2in} +\sum_{m'_2 \neq m_2} \frac{  P^{\otimes n}(z^n|x^n_1(m_1))}{2^{n(R_1+R_2)} Q_Z^{\otimes n}(z^n)} +1 \bigg) \label{eq:bound_D}
\end{align}
where $(a)$ follows by Jensen's inequality. We finally write the right-hand side of~\eqref{eq:bound_D} as $\Psi_1 + \Psi_2$ with
\begin{align*}
\Psi_1& \triangleq \frac{1}{2^{n(R_1+R_2)}} \sum_{m_1} \sum_{m_2} \twocolbreak
\hspaceonetwocol{0in}{0.2in}  \sum_{(x_1^n,x_2^n,z^n) \in \mathcal{T}_\epsilon^n(P_{X_1,X_2,Z})} P^{\otimes n} (x^n_1(m_1),x^n_2(m_1,m_2),z^n)\nonumber\\
&\hspaceonetwocol{0.1in}{0.2in} \log \bigg( \frac{W^{\otimes n}(z^n|x^n_1(m_1),x^n_2(m_1,m_2))}{2^{n(R_1+R_2)} Q_Z^{\otimes n}(z^n)} \twocolbreak
 \hspaceonetwocol{0in}{0.2in} +\sum_{m'_2 \neq m_2} \frac{  P^{\otimes n}(z^n|x^n_1(m_1))}{2^{n(R_1+R_2)} Q_Z^{\otimes n}(z^n)} +1 \bigg)\nonumber\\
&  \leq \log \Big( \frac{2^{-n(1-\epsilon)H(Z|X_1,X_2)}}{2^{n(R_1+R_2)} 2^{-n(1+\epsilon)H(Z)}}  \twocolbreak
\hspaceonetwocol{0in}{0.2in}  + \frac{2^{nR_2} 2^{-n(1-\epsilon)H(Z|X_1)}}{2^{n(R_1+R_2)} 2^{-n(1+\epsilon)H(Z)}} +1 \Big) \nonumber\\
& \leq \log \Big( 2^{-n(R_1+R_2-I(X_1,X_2;Z)-2\epsilon H(Z))} \twocolbreak
\hspaceonetwocol{0in}{0.2in}  + 2^{-n(R_1-I(X_1;Z)-2\epsilon H(Z))} + 1  \Big) \nonumber
\end{align*}
and 
\begin{align*}
\Psi_2& \triangleq \frac{1}{2^{n(R_1+R_2)}} \sum_{m_1} \sum_{m_2} \twocolbreak
\hspaceonetwocol{0in}{0.2in}  \sum_{(x_1^n,x_2^n,z^n) \notin \mathcal{T}_\epsilon^n(P_{X_1,X_2,Z})} P^{\otimes n} (x^n_1(m_1),x^n_2(m_1,m_2),z^n)\nonumber\\
&\hspaceonetwocol{0.1in}{0.2in}  \log \bigg( \frac{W^{\otimes n}(z^n|x^n_1(m_1),x^n_2(m_1,m_2))}{2^{n(R_1+R_2)} Q_Z^{\otimes n}(z^n)} \twocolbreak
 \hspaceonetwocol{0in}{0.2in}  +\sum_{m'_2 \neq m_2} \frac{   P^{\otimes n}(z^n|x^n_1(m_1))}{2^{n(R_1+R_2)} Q_Z^{\otimes n}(z^n)} +1 \bigg)\nonumber\\
& \leq 2 |\mathcal{X}_1||\mathcal{X}_2||\mathcal{Z}| e^{-n\epsilon^2  \mu_{X_1X_2Z}} n \log (\frac{2}{\mu_Z} +1) \nonumber
\end{align*}
 where
 \begin{align}
\mu_Z &=  \min_{\substack{z \in \mathcal{Z}\\\text{s.t. } Q(z)>0}} Q(z) \nonumber\\
\mu_{X_1 X_2 Z} &= \min_{\substack{(x_1,x_2,z) \in (\mathcal{X}_1,\mathcal{X}_2,\mathcal{Z}) \\ \text{s.t. } Q(x_1,x_2,z)>0 } } Q(x_1,x_2,z) \nonumber
\end{align}
Combining the bounds on $\Psi_1$ and $\Psi_2$ we obtain $\mathbb{E}(\mathbb{D}(P_{Z^n}||Q_Z^{\otimes n})) \rightarrow 0$ exponentially with $n$ if $R_1>I(X_1;Z)+2\epsilon H(Z)$ and $R_1+R_2>I(X_1,X_2;Z)+2\epsilon H(Z)$. This implies, by  Markov's inequality, that $\text{Pr}(\mathbb{D}(P_{Z^n}||Q_Z^{\otimes n})>\eta_n)\xrightarrow[]{n\to \infty} 0$ for a suitable choice of $\eta_n$;  $\eta_n=e^{-n\alpha}$ for $\alpha>0$.


\subsubsection{Converse} \hfil\\
By assumption, 
\begin{align*}
 \epsilon &\geq \mathbb{D}(P_{Z^n}||Q_Z^{\otimes n}) \\
 &=\sum_{z^n} P(z^n)\log\frac{P(z^n)}{Q_Z^{\otimes n}(z^n)}\\
 &=\sum_{i=1}^n \bigg(\sum_{z_i}P_Z(z_i)\log \frac{1}{Q(z_i)}- H(Z_i|Z^{i-1}) \Big)\\
 & \stackrel{(a)}{\geq} \sum_{i=1}^n \bigg(\sum_{z_i}P(z_i)\log \frac{1}{Q(z_i)}- H(Z_i)\bigg)\\
 &=\sum_{i=1}^n \mathbb{D} (P_{Z_i}||Q_Z)\\
 & \stackrel{(b)}{\geq} n \mathbb{D}(\tilde{P}_Z||Q_Z)
\end{align*}
where (a) follows because conditioning does not increase entropy and (b) follows by Jensen's Jensen's inequality and the convexity of $\mathbb{D}(\cdot||\cdot)$ with $\tilde{P}_Z(z)\triangleq \frac{1}{n}\sum_{i=1}^n P_{Z_i}(z)$. First, note that

\begin{align}
n R_1 &= H(M_1) \label{conv1}\\
&\geq I(M_1;Z^n) \nonumber\\
&\stackrel{(a)}{=}  I(M_1,X_1^n;Z^n) \nonumber\\
&\geq I(X_1^n;Z^n) \nonumber\\
&= I(X_1^n,X_2^n;Z^n) - I(X_2^n;Z^n|X_1^n) \nonumber\\
&\stackrel{(b)}{\geq}  \sum_{x_1^n} \sum_{x_2^n} \sum_{z^n} P(x_1^n,x_2^n,z^n) \log \frac{W^{\otimes n}(z^n|x_1^n,x_2^n)}{P_{Z^n}(z^n)}  \twocolbreak
 \hspaceonetwocol{0.0in}{0.2in}-\sum_i I(X_{2i};Z_i|X_{1i}) \nonumber\\
&= \sum_{x_1^n} \sum_{x_2^n} \sum_{z^n} P(x_1^n,x_2^n,z^n) \log \frac{W^{\otimes n}(z^n|x_1^n,x_2^n)}{Q_Z^{\otimes n}(z^n)}  \twocolbreak
 \hspaceonetwocol{0.0in}{0.2in}-\mathbb{D}(P_{Z^n}||Q_Z^{\otimes n}) - \sum_i I(X_{2i};Z_i|X_{1i}) \nonumber\\
&\geq \sum_{x_1^n} \sum_{x_2^n} \sum_{z^n} P(x_1^n,x_2^n,z^n) \log \frac{W^{\otimes n}(z^n|x_1^n,x_2^n)}{Q_Z^{\otimes n}(z^n)} \twocolbreak
\hspaceonetwocol{0.0in}{0.2in}- \sum_i I(X_{2i};Z_i|X_{1i}) - \epsilon \nonumber\\
&= \sum_i \sum_{x_{1i}} \sum_{x_{2i}} \sum_{z_i} P(x_{1i},x_{2i},z_i)\bigg( \log \frac{W(z_i|x_{1i},x_{2i})}{Q(z_i)} \twocolbreak
 \hspaceonetwocol{0.0in}{0.2in}- \log \frac{W(z_i|x_{1i},x_{2i})}{P(z_i|x_{1i})} \bigg) -\epsilon\nonumber\\
&= \sum_i  \sum_{x_{1i}} \sum_{x_{2i}} \sum_{z_i} P(x_{1i},x_{2i},z_i) \log \frac{P(z_i|x_{1i})}{Q(z_i)} -\epsilon\nonumber\\
&=\sum_i  \sum_{x_{1i}} \sum_{z_i} P(x_{1i},z_i) \log \frac{P(z_i|x_{1i})}{Q(z_i)} -\epsilon \nonumber\\
&=\sum_i \mathbb{D}(P_{X_{1i}Z_i}||P_{X_{1i}} Q_{Z_i}) -\epsilon \nonumber\\
&\stackrel{(c)}{\geq}  n \mathbb{D}\bigg(\frac{\sum_i P_{X_{1i}Z_i} }{n} \bigg|\bigg| \frac{\sum_i P_{X_{1i}}}{n}  Q_{Z} \bigg) -\epsilon \nonumber\\  
&\stackrel{(d)}{=} n \mathbb{D}(\tilde{P}_{X_1,Z}||\tilde{P}_{X_1} Q_Z) -\epsilon \nonumber\\
&= n \sum_{x_1} \sum_{z} \tilde{P}_{X_1,Z}(x_1,z) \log \frac{\tilde{P}_{X_1,Z}(x_1,z)}{\tilde{P}_{X_1}(x_1) Q_{Z}(z)} -\epsilon \nonumber\\
&= n \sum_{x_1} \sum_{z} \tilde{P}_{X_1,Z}(x_1,z) \log \frac{\tilde{P}_{X_1,Z}(x_1,z)}{\tilde{P}_{X_1}(x_1) \tilde{P}_{Z}(z)} \twocolbreak
 \hspaceonetwocol{0.0in}{0.2in}  + n \sum_{x_1} \sum_{z} \tilde{P}_{X_1,Z}(x_1,z) \log \frac{\tilde{P}_{Z}(z)}{Q_{Z}(z)} -\epsilon \nonumber\\
&= n I(\tilde{X}_1;\tilde{Z}) + n \mathbb{D}( \tilde{P}_{Z} || Q_{Z}) -\epsilon \nonumber\\
&\geq n I(\tilde{X}_1; \tilde{Z}) -\epsilon \label{conv2}
\end{align}
where 
\begin{enumerate}[(a)]
\item follows from the definition of the deterministic encoding functions in~(\ref{encfn_degraded});
\item follows because conditioning does not increase entropy and the channel is discrete memoryless, therefore $I(X_2^n;Z^n|X_1^n) =\sum H(Z_i|Z^{i-1},X_1^n)-H(Z_i|Z^{i-1},X_1^n,X_2^n) \leq \sum H(Z_i|X_{1i})-H(Z_i|X_{1i},X_{2i})   \leq \sum_{i=1}^n I(X_{2i};Z_{i}|X_{1i})$;
\item follows by Jensen's inequality and the convexity of $\mathbb{D}(\cdot||\cdot)$;
\item follows by defining $\tilde{P}_{X_1,Z}(x_1,z) \triangleq \frac{1}{n} \sum_i P_{X_{1i},Z_i}(x_1,z)$ and  $\tilde{P}_{X_1}(x_1) \triangleq \frac{1}{n} \sum_i P_{X_{1i}}(x_1)$ where $\tilde{P}_{X_1,X_2}(x_1,x_2) \triangleq \frac{1}{n} \sum_i P_{X_{1i},X_{2i}}(x_1,x_2)$, $\tilde{P}_{X_1,X_2,Z}(x_1,x_2,z)\triangleq \frac{1}{n} \sum_i P_{X_{1i},X_{2i},Z_i}(x_1,x_2,z) = W_{Z|X_1,X_2}(z|x_1,x_2) \tilde{P}_{X_1,X_2}(x_1,x_2)$ and $\tilde{P}_{X_1,Z}(x_1,z)=\sum_{x_2}\tilde{P}_{X_1,X_2,Z}(x_1,x_2,z)$ .
\end{enumerate}

Next, observe that
\begin{align}
n (&R_1+R_2) \nonumber\\ 
&= H(M_1,M_2) \label{conv3}\\
&\geq I(M_1,M_2;Z^n) \nonumber\\
&\geq I(X_1^n,X_2^n;Z^n) +\mathbb{D}(P_{Z^n} ||Q_Z^{\otimes n}) - \epsilon \nonumber\\
&= \sum_{x_1^n} \sum_{x_2^n} \sum_{z^n} P(x_1^n,x_2^n,z^n) \log \frac{P(x_1^n,x_2^n,z^n)}{P(x_1^n,x_2^n)P_{Z^n}(z^n)}\twocolbreak
 + \sum_{z^n} P(z^n) \log \frac{P_{Z^n}(z^n)}{Q_Z^{\otimes n}(z^n)} - \epsilon \nonumber\\
&= \sum_{x_1^n} \sum_{x_2^n} \sum_{z^n} P(x_1^n,x_2^n,z^n) \log \frac{W^{\otimes n}(z^n|x_1^n,x_2^n)}{Q_Z^{\otimes n}(z^n)} - \epsilon \nonumber\\
&= \sum_i \sum_{x_{1i}} \sum_{x_{2i}} \sum_{z_i} P(x_{1i},x_{2i},z_i) \log \frac{W(z_i|x_{1i},x_{2i})}{Q(z_i)} -\epsilon \nonumber\\
&= \sum_i \mathbb{D}(P_{X_{1i},X_{2i},Z_i}||P_{X_{1i},X_{2i}}Q_{Z_i}) -\epsilon \nonumber\\
&\stackrel{(a)}{\geq} n \mathbb{D}\bigg(\frac{\sum_i P_{X_{1i},X_{2i},Z_i} }{n} \bigg|\bigg| \frac{\sum_i P_{X_{1i},X_{2i}}}{n}  Q_{Z} \bigg) -\epsilon \nonumber\\
&\stackrel{(b)}{=} n \mathbb{D}(\tilde{P}_{X_{1},X_{2},Z}||\tilde{P}_{X_{1},X_{2}}  Q_{Z}) -\epsilon \nonumber\\
&= n \mathbb{D}(\tilde{P}_{X_{1},X_{2},Z}||\tilde{P}_{X_{1},X_{2}}  \tilde{P}_{Z}) + n \mathbb{D}(\tilde{P}_{Z}||  Q_{Z})  -\epsilon \nonumber\\
&\geq n I(\tilde{X}_{1},\tilde{X}_{2};\tilde{Z})-\epsilon \label{conv4} 
\end{align}
where
\begin{enumerate}[(a)]
\item follows by Jensen's inequality and the convexity of $\mathbb{D}(\cdot||\cdot)$;
\item follows by defining $\tilde{P}_{X_1,X_2,Z}(x_1,x_2,z)\triangleq \frac{1}{n} \sum_i P_{X_{1i},X_{2i},Z_i}(x_1,x_2,z)$ and $\tilde{P}_{X_1,X_2}(x_1,x_2)\triangleq \frac{1}{n} \sum_i P_{X_{1i},X_{2i}}(x_1,x_2)$ with $\tilde{P}_{X_1,X_2,Z}(x_1,x_2,z) = W_{Z|X_1,X_2}(z|x_1,x_2) \tilde{P}_{X_1,X_2}(x_1,x_2,z)$.
\end{enumerate}

The final step of this converse proof, and other converse proofs in this paper, is to show the continuity of the resolvability region at $\epsilon \to 0$. For a proof of this statement, we refer the reader to \cite[Section VI.C]{synthesis}, which can be extended to a \ac{MAC}.  


\subsection{Theorem~\ref{th_noncausal}: \ac{MAC} with Non-Causal Cribbing} \label{noncausal}

\subsubsection{Achievability} \hfill\\
Consider a distribution  $P(x_1,x_2)=P(x_1)P(x_2|x_1)$ such that $\sum_{x_1,x_2} P(x_1,x_2)W(z|x_1,x_2)=Q_Z(z)$. 
\begin{itemize}
\item Independently generate $2^{nR_1}$ codewords $x_1^n$ each with probability $P(x_1^n)=P_{X_1}^{\otimes n}(x_1^n)$. Label them $x_1^n(m_1)$, $m_1 \in \llbracket 1,2^{nR_1}\rrbracket$.
\item For every $x_1^n(m_1)$, independently generate $2^{nR_2}$ codewords $x_2^n$ each with probability $P(x_2^n|x_1^n(m_1))=P_{X_2|X_1}^{\otimes n}(x_2^n|x_1^n(m_1))$. Label them $x_{2}^n(x_1^n(m_1),m_2)$,  $m_{2} \in \llbracket 1,2^{nR_2}\rrbracket$.
\end{itemize}

Unlike the case of degraded message sets, here the codewords $x_2^n$ are a function of $x_1^n(m_1)$ instead of $m_1$. For $m_1 \in \llbracket 1,2^{nR_1}\rrbracket$ and $m_{2} \in \llbracket 1,2^{nR_2}\rrbracket$, let $X^n_1(m_1)$ and $X^n_2(X_1^n(m_1),m_2)$ denote  random variables representing the randomly generated codewords. The average \ac{KL} divergence is:
\begin{align}
&\mathbb{E}\big(\mathbb{D}(P_{Z^n}||Q_Z^{\otimes n})\big) \nonumber\\
&= \mathbb{E}\big(\sum_{z^n} P_{Z^n}(z^n) \log \frac{P_{Z^n}(z^n)}{Q_Z^{\otimes n}(z^n)}\big) \nonumber\\
&= \mathbb{E}\bigg(\sum_{z^n} \frac{\sum_{m_1,m_2} W^{\otimes n}(z^n|X^n_1(m_1),X^n_2(X^n_1(m_1),m_2))}{2^{n(R_1+R_2)}} \nonumber\\
 &\hspaceonetwocol{2.5in}{0.2in}  \log \frac{\sum_{m'_1,m'_2} W^{\otimes n}(z^n|X^n_1(m'_1),X^n_2(X^n_1(m'_1),m'_2))}{2^{n(R_1+R_2)} Q_Z^{\otimes n}(z^n)}\bigg) \nonumber\\
&\stackrel{(a)}{\leq} \frac{1}{2^{n(R_1+R_2)}} \sum_{m_1,m_2} \sum_{z^n} \sum_{x_1^n(m_1)} \sum_{x_2^n(x_1^n(m_1),m_2)} \twocolbreak
\hspaceonetwocol{0in}{0.2in} P(x_1^n(m_1),x_2^n(x_1^n(m_1),m_2),z^n)  \nonumber\\
&\hspaceonetwocol{2.5in}{0.2in}  \log \mathbb{E}_{\setminus (m_1,m_2)} \frac{\sum\limits_{m'_1,m'_2} \!\! W^{\otimes n}(z^n|X^n_1(m'_1),X^n_2(X^n_1(m'_1),m'_2))}{2^{n(R_1+R_2)} Q_Z^{\otimes n}(z^n) }\nonumber\\
&= \frac{1}{2^{n(R_1+R_2)}} \sum_{m_1,m_2} \sum_{z^n} \sum_{x_1^n(m_1)} \sum_{x_2^n(x_1^n(m_1),m_2)} \twocolbreak
\hspaceonetwocol{0in}{0.2in}  P(x_1^n(m_1),x_2^n(x_1^n(m_1),m_2),z^n)  \nonumber\\
 &\hspaceonetwocol{0.2in}{0.2in}   \log \mathbb{E}_{\setminus(m_1,m_2)} \Bigg( \frac{W^{\otimes n}(z^n|x^n_1(m_1),x^n_2(x^n_1(m_1),m_2))}{2^{n(R_1+R_2)} Q_Z^{\otimes n}(z^n)} \nonumber\\
 &\hspaceonetwocol{2.2in}{0.2in}   \!\!\! + \!\!\! \!\! \sum_{m_1' \neq m_1} \!\!\!\!\! \Big[ \mathds{1}_{\{ x_1^n(m_1')=x_1^n(m_1) \}} \frac{W^{\otimes n}(z^n|X^n_1(m_1'),x^n_2(X^n_1(m_1'),m_2))}{2^{n(R_1+R_2)} Q_Z^{\otimes n}(z^n)}\nonumber\\
&\hspaceonetwocol{2.2in}{0.2in}   +  \mathds{1}_{\{ x_1^n(m_1') \neq x_1^n(m_1)\} } \frac{W^{\otimes n}(z^n|X^n_1(m_1'),x^n_2(X^n_1(m_1'),m_2))}{2^{n(R_1+R_2)} Q_Z^{\otimes n}(z^n)} \Big] \nonumber\\
&\hspaceonetwocol{2.2in}{0.2in}  + \sum_{m_2' \neq m_2} \frac{W^{\otimes n}(z^n|x^n_1(m_1),X^n_2(x^n_1(m_1),m_2'))}{2^{n(R_1+R_2)} Q_Z^{\otimes n}(z^n)}\nonumber\\
 &\hspaceonetwocol{2.2in}{0.2in} + \sum_{\substack{m_2' \neq m_2 \\ m_1' \neq m_1 }} \frac{W^{\otimes n}(z^n|X^n_1(m_1'),X^n_2(X^n_1(m_1'),m_2'))}{2^{n(R_1+R_2)} Q_Z^{\otimes n}(z^n)} \Bigg)  \nonumber\\
&\leq \frac{1}{2^{n(R_1+R_2)}} \sum_{m_1,m_2} \sum_{z^n} \sum_{x_1^n(m_1)} \sum_{x_2^n(x_1^n(m_1),m_2)} \twocolbreak
\hspaceonetwocol{0in}{0.2in}  P(x_1^n(m_1),x_2^n(x_1^n(m_1),m_2),z^n)  \nonumber\\
 &\hspaceonetwocol{0.2in}{0.2in}   \log  \Bigg( \frac{W^{\otimes n}(z^n|x^n_1(m_1),x^n_2(x^n_1(m_1),m_2))}{2^{n(R_1+R_2)} Q_Z^{\otimes n}(z^n)} \nonumber\\
&\hspaceonetwocol{1.5in}{0.2in}  + \mathbb{E}_{\setminus(m_1,m_2)} 
\Big( \sum_{m_1' \neq m_1}  \twocolbreak
 \hspaceonetwocol{0in}{0.2in} \Big[ \mathds{1}_{\{ x_1^n(m_1')=x_1^n(m_1)\} } \frac{W^{\otimes n}(z^n|X^n_1(m_1'),X^n_2(X^n_1(m_1'),m_2))}{2^{n(R_1+R_2)} Q_Z^{\otimes n}(z^n)}\nonumber\\
& \hspaceonetwocol{1.5in}{0.2in}   +  \mathds{1}_{\{ x_1^n(m_1') \neq x_1^n(m_1)\} } \frac{W^{\otimes n}(z^n|X^n_1(m_1'),X^n_2(X^n_1(m_1'),m_2))}{2^{n(R_1+R_2)} Q_Z^{\otimes n}(z^n)} \Big] \Big) \nonumber\\
& \hspaceonetwocol{1.5in}{0.2in}   + \sum_{m_2' \neq m_2} \frac{ P^{\otimes n}(z^n|x^n_1(m_1))}{2^{n(R_1+R_2)} Q_Z^{\otimes n}(z^n)}+1 \Bigg)  \label{eqn5}\\
&\stackrel{(b)}{=} \Psi_1+\Psi_2 \nonumber
\end{align}
where
\begin{enumerate}[(a)]
\item follows by Jensen's inequality;
\item similar to the \ac{MAC} with degraded message sets, $\Psi_1$ is taking the summation $\sum_{x_1^n,x_2^n,z^n}$ in \eqref{eqn5} over $(x_1^n,x_2^n,z^n) \in \mathcal{T}_{\epsilon}^{n}(P_{X_1,X_2,Z})$  and $\Psi_2$ is taking the same summation over $(x_1^n,x_2^n,z^n) \not\in \mathcal{T}_{\epsilon}^{n}(P_{X_1,X_2,Z})$.
\end{enumerate}

Hence,
\begin{align}
\Psi_1&\leq \log \Big( 2^{-n(R_1+R_2)} 2^{-n(1-\epsilon) H(Z|X_1,X_2)} 2^{n(1+\epsilon) H(Z)} \twocolbreak
\includeonetwocol{}{\hspace{0.2in}} + 2^{-nR_2} 2^{-n(1-\epsilon)(H(X_1)+H(Z|X_1,X_2))} 2^{n(1+\epsilon) H(Z) }\nonumber\\
 &\hspaceonetwocol{0.2in}{0.2in}  +2^{-nR_2}+ 2^{-nR_1} 2^{-n(1-\epsilon)(H(Z|X_1))} 2^{n(1+\epsilon) H(Z) } +1 \Big)
\end{align}

Now $\mathbb{E}(\mathbb{D}(P_{Z^n}||Q_Z^{\otimes n})) \xrightarrow[]{n\to \infty} 0$ if $R_1>I(X_1;Z)+2\epsilon H(Z)$, $R_2>I(X_1,X_2;Z)-H(X_1)+2\epsilon H(Z)$ and $R_1+R_2>I(X_1,X_2;Z)+2\epsilon H(Z)$.

\subsubsection{Converse} \hfill\\
By assumption similar to the degraded message set \ac{MAC}, $\epsilon \geq \mathbb{D}(P_{Z^n}||Q_Z^{\otimes n}) \geq  n\mathbb{D}(\tilde{P}_Z||Q_Z)$ with $\tilde{P}_Z(z)\triangleq \frac{1}{n}\sum_{i=1}^n P_{Z_i}(z)$. The proofs for the rate $R_1$ and the sum rate $R_1+R_2$ are similar to the \ac{MAC} with degraded message sets case by repeating the steps in \eqref{conv1}-\eqref{conv2} and \eqref{conv3}-\eqref{conv4} respectively. So we will proceed by presenting the proof for the rate $R_2$,

Note that
\begin{align}
n R_2 &= H(M_2) \nonumber\\
&\geq H(M_2|X_1^n)\nonumber\\
&\geq I(M_2;Z^n|X^n_1) \nonumber\\
&\stackrel{(a)}{=} I(M_2,X_2^n;Z^n|X_1^n)\nonumber\\
&\geq I(X_2^n;Z^n|X_1^n)\nonumber\\
&= I(X_1^n,X_2^n;Z^n) - I(X_1^n;Z^n)\nonumber\\
&= \sum_{x_1^n} \sum_{x_2^n} \sum_{z^n} P(x_1^n,x_2^n,z^n) \log \frac{P(x_1^n,x_2^n,z^n)}{P(x_1^n,x_2^n)P_{Z^n}(z^n)}  \twocolbreak
 \hspaceonetwocol{0.0in}{0.2in}- H(X_1^n) + H(X_1^n|Z^n)\nonumber\\
&\geq \sum_{x_1^n} \sum_{x_2^n} \sum_{z^n} P(x_1^n,x_2^n,z^n) \log \frac{W^{\otimes n}(z^n|x_1^n,x_2^n)}{Q_Z^{\otimes n}(z^n)} \twocolbreak
 \hspaceonetwocol{0.0in}{0.2in}-\mathbb{D}(P_{Z^n}||Q_Z^{\otimes n}) -H(X_1^n)\nonumber\\
& \geq \sum_i \sum_{x_{1i}} \sum_{x_{2i}} \sum_{z_i} P(x_{1i},x_{2i},z_i) \log \frac{W(z_i|x_{1i},x_{2i})}{Q(z_i)}  \twocolbreak
 -\sum_i H(X_{1i}) -\epsilon \nonumber\\
&= \sum_i \mathbb{D}(P_{X_{1i},X_{2i},Z_i}||P_{X_{1i},X_{2i}}Q_{Z_i}) -\sum_i H(X_{1i})-\epsilon\nonumber\\
&\stackrel{(b)}{\geq} n\mathbb{D}(\tilde{P}_{X_{1},X_{2},Z}||\tilde{P}_{X_{1},X_{2}}  Q_{Z}) - n H(\tilde{X}_1) -\epsilon \nonumber\\
&= n\mathbb{D}(\tilde{P}_{X_{1},X_{2},Z}||\tilde{P}_{X_{1},X_{2}}  \tilde{P}_{Z}) + \mathbb{D}(\tilde{P}_{Z}||  Q_{Z})\twocolbreak
\hspaceonetwocol{0.0in}{0.2in}- n H(\tilde{X}_1) -\epsilon \nonumber\\
&\geq n I(\tilde{X}_{1},\tilde{X}_{2};\tilde{Z})- n H(\tilde{X}_1) -\epsilon \nonumber
\end{align}
where
\begin{enumerate}[(a)]
\item follows from the definition of the deterministic encoding functions in~(\ref{encfn_noncausal});
\item follows by Jensen's inequality, the convexity of $\mathbb{D}(\cdot||\cdot)$, concavity of $H(\cdot)$ and defining $\tilde{P}_{X_1,X_2,Z}(x_1,x_2,z)\triangleq \frac{1}{n} \sum_i P_{X_{1i},X_{2i},Z_i}(x_1,x_2,z)$ and $\tilde{P}_{X_1,X_2}(x_1,x_2)\triangleq \frac{1}{n} \sum_i P_{X_{1i},X_{2i}}(x_1,x_2)$ with $\tilde{P}_{X_1,X_2,Z} (x_1,x_2,z)= W_{Z|X_1,X_2}(z|x_1,x_2) \tilde{P}_{X_1,X_2}(x_1,x_2)$.
\end{enumerate}

To conclude the converse proof, we again refer the reader to~\cite[Section VI.C]{synthesis} for the continuity of the resolvability region at $\epsilon \to 0$.


\subsection{Theorem~\ref{th_strict}: \ac{MAC} with Strictly-Causal Cribbing}
\label{strictlycausal}

\subsubsection{Achievability} \hfil\\
To handle the strict causality constraint, we adopt a block-Markov encoding scheme over $B > 0$ consecutive and dependent blocks, each consisting of $r$ transmissions such that $n = rB$. The vector of $n$ channel outputs $Z^n$ at the channel output may then be described as $Z^n \triangleq (Z^r_1,\cdots,Z^r_B)$, where each $Z^r_b$ for $b \in \llbracket 1,B \rrbracket$ describes the observations in block $b$. The distribution induced by the coding scheme is the joint distribution $P_Z^n \triangleq P_{Z^r_1,\cdots,P_{Z^r_B}}$, while the target output distributions is a product distribution of product distributions $Q^{\otimes n}_Z \triangleq \prod_{j=1}^B Q_Z^{\otimes r}$ . Notice that
 
\begin{align}
\mathbb{D}(P_{Z^n}||Q_Z^{\otimes n}) &= \mathbb{D}(P_{Z_1^r...Z_B^r}||Q_{Z}^{\otimes rB}) \nonumber\\
&= \sum_{j=1}^{B} \mathbb{D}(P_{Z_j^r|Z_{j+1}^{B,r}}||Q_Z^{\otimes r}|P_{Z_{j+1}^{B,r}}) \nonumber\\
&= \sum_{j=1}^{B} \mathbb{D}(P_{Z_j^r}||Q_Z^{\otimes r}) + \sum_{j=1}^{B} \mathbb{D}(P_{Z_j^r|Z_{j+1}^{B,r}}||P_{Z_j^r}|P_{Z_{j+1}^{B,r}}) \nonumber\\
&= \sum_{j=1}^{B} \mathbb{D}(P_{Z_j^r}||Q_Z^{\otimes r}) + \sum_{j=1}^{B} I(Z_j^r;Z_{j+1}^{B,r}) \label{eq6}
\end{align}
where $Z_{j+1}^{B,r}=\{ Z_{j+1}^r,\dots Z_{B}^r\}$. This suggests that to achieve $\mathbb{D}(P_Z^n||Q_Z^{\otimes n})~\xrightarrow[]{n\rightarrow\infty}~0$, we
may enforce $\forall j=\{1,...,B\}$ both $\mathbb{D}(P_{Z_j^r}||Q_Z^{\otimes r})~\xrightarrow[]{r\rightarrow\infty}~0$ and $I(Z_j^r;Z_{j+1}^{B,r})~\xrightarrow[]{r\rightarrow\infty}~0$ sufficiently fast with $r$. 
As shown next, we achieve this by constructing a code that approximates $Q_Z^{\otimes r}$ in every block and show that the dependencies across blocks, created by block-Markov coding, can be eliminated by suitably recycling randomness from one block to the next.

\begin{figure}
\begin{center}
    \includegraphics[width=0.9\textwidth]{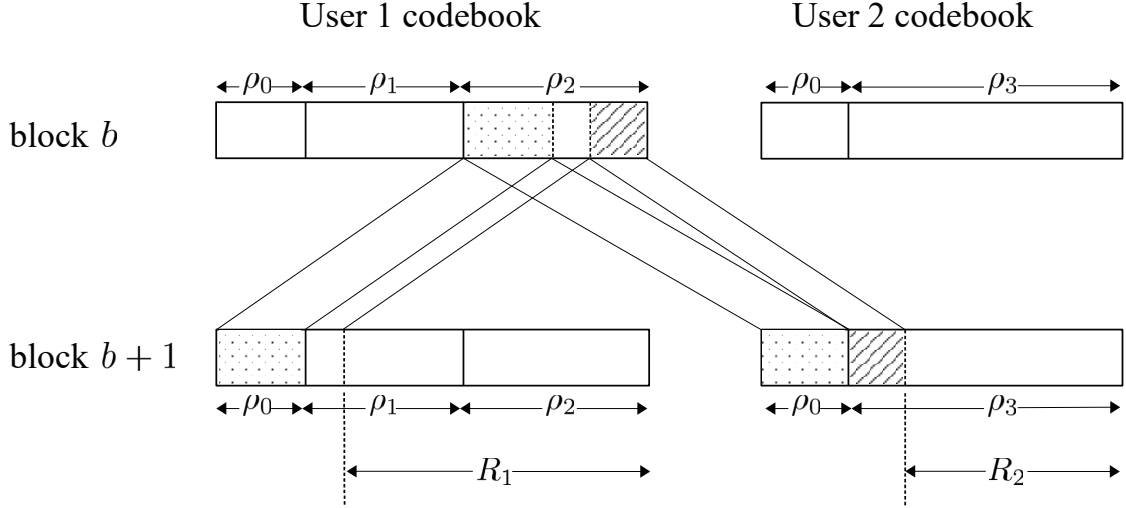}
\end{center}
\caption{Codebook structure at both users.} 
\label{codebook_structure}
\end{figure}

Consider a distribution $P_{UX_1X_2Z} = P_U P_{X_1|U} P_{X_2|U} W_{Z|X_1X_2}$ with marginal $Q_Z$ that satisfies $H(X_1|U)>I(UX_1;Z)$. For every block $b \in \llbracket 1,B \rrbracket$:

\begin{itemize}
\item Independently generate $2^{r \rho_0}$ codewords according to $P_U^{\otimes r}$ and label them $u^r(m_0^{(b)})$, where $m_0^{(b)} \in \llbracket 1,2^{r \rho_0} \rrbracket$.
\item For every $m_0^{(b)}$, independently generate $2^{r(\rho_1+\rho_2)}$ codewords according to $\prod_{i=1}^r P_{X_1|U=u_i^r(m_0^{(b)})}$; label them $x^r_1(m_0^{(b)},m'^{(b)}_1,m''^{(b)}_1)$ for $m'^{(b)}_1 \in \llbracket 1,2^{r \rho_1} \rrbracket $ and $m''^{(b)}_1 \in \llbracket 1,2^{r \rho_2} \rrbracket $. Note that $m_1^{(b)}=(m'^{(b)}_1,m''^{(b)}_1)$.
\item For every $m_0^{(b)}$, independently generate $2^{r\rho_3}$ codewords according to $\prod_{i=1}^r P_{X_2|U=u_i^r(m_0^{(b)})}$; label them $x^r_2(m_0^{(b)},m_2^{(b)})$, $m_2^{(b)} \in \llbracket 1,2^{r \rho_3} \rrbracket $.
\end{itemize}
 
 In every block $b\in \llbracket 1,B \rrbracket$, assuming all messages are chosen independently and uniformly in their respective sets and that both encoders use the same $M^{(b)}_0$, our objective is to establish conditions for $\rho_0$, $\rho_1$, $\rho_2$ and $\rho_3$ such that:
 \begin{enumerate}
 \item messages $M'^{(b)}_1$ and $M''^{(b)}_1$ can be decoded from $X^r_1$ knowing $M^{(b)}_0$;
 \item message $M''^{(b)}_1$ is nearly independent of $Z^r_b$;
 \item the distribution induced by the code, which we denote $\bar{P}_{Z^r_b}$ to indicate that $M_0^{(b)}$ is assumed known to both encoders, approximates $Q_Z^{\otimes r}$.
 \end{enumerate}
The message $M''^{(b)}_1$ is the part of $M^{(b)}_1$ that we wish to recycle toward the creation of $M_0^{(b+1)}$, which itself constitutes the cooperating message between the two encoders. This dependency between blocks can be hidden at the output of the channel by transmitting $M''^{(b)}_1$ securely over the wiretap channel with $X_1$ as its input, the second encoder as the legitimate receiver, and  $Z$ as the wiretapper. If we let $P_e^{(b)}$ denote the average probability of error for decoding $M'^{(b)}_1$ and $M''^{(b)}_1$ from $(X_1^r,M_0^{(b)})$, a standard argument shows that, when averaging over the randomly generated codes, $\mathbb{E}\left(P_e^{(b)}\right) < 2^{-\alpha r}$ for some $\alpha>0$ and all $r$ large enough if
 \begin{align}
 \rho_1+\rho_2<H(X_1|U). \label{rho1}
 \end{align}
 If we let $D^{(b)}\eqdef \mathbb{D}(\bar{P}_{Z_b^r M''^{(b)}_1} || Q_Z^{\otimes r} \bar{P}_{M''^{(b)}_1})$, a standard argument~(see~\cite[Section III]{strong_coord} for a similar result with variational distance and Appendix \ref{appendix:resolv_rates} for more detailed steps) shows that, when averaging over the randomly generated codes, $\mathbb{E} \left(D^{(b)} \right) < 2^{-\beta r}$ for some $\beta > 0$ and all $r$ large enough if
 \begin{align}
 \rho_0 &>I(U;Z),\label{rho00}\\
 \rho_0+\rho_1 &> I(UX_1;Z), \label{rho01}\\
 \rho_0+\rho_1+\rho_3 &> I(X_1X_2;Z), \label{rho02}\\
 \rho_0+\rho_3 &>I(UX_2;Z). \label{rho03}
 \end{align}
 
 Let $\epsilon >0$ and set
 \begin{align}
 \rho_0 &= I(U;Z)+\epsilon, \label{rho3}\\
 \rho_1 &= I(X_1;Z|U)+\epsilon,\\
 \rho_2 &= H(X_1|U) - I(X_1;Z|U) -2 \epsilon, \\
 \rho_3 &= I(X_2;Z|UX_1) +\epsilon, \label{rho4}
 \end{align}
 which is compatible with constraints \eqref{rho1}-\eqref{rho01}. The choice is also compatible with \eqref{rho02} because $I(UX_1X_2;Z)=I(X_1X_2;Z)$. The choice is finally compatible with \eqref{rho03} because $I(X_2;Z|UX_1)=H(X_2|UX_1)-H(X_2|UX_1Z)=H(X_2|U)-H(X_2|UX_1Z)\geq H(X_2|U)-H(X_2|UZ)=I(X_2;Z|U)$. Hence, by an expurgation argument, for every $b\in \llbracket 1,B \rrbracket$ there exists a code for block b such that 
 \begin{align}
 P_e^{(b)} < 2^{-\alpha' r} \quad \text{and} \quad D^{(b)}<2^{-\beta' r} \label{v3}
 \end{align}
 for some $\alpha'$, $\beta'>0$ and all $r$ large enough.

 The codes hence obtained are chained across $B$ blocks as follows. In Block~1, we assume that the encoders have access to a common message $M_0^{(1)}$ through some private common randomness (see Remark~\ref{Remark:noncooperative} for justification). In block $b>1$, we assume for now that Encoder~2 knows $M_0^{(b)}$. It is then able to form estimates $\widehat{M}'^{(b)}_1, \widehat{M}''^{(b)}_1$, which are correct with high probability. Since we have assumed that $H(X_1|U)>I(UX_1;Z)$, we have $\rho_2>\rho_0$ and an amount $\rho_0$ of the rate $\rho_2$ may be \textit{recycled} toward the creation of $M_0^{(b+1)}$ (see Fig.~\ref{codebook_structure}). Furthermore, for $\gamma \in \llbracket 0,1 \rrbracket$, an amount $\gamma(\rho_2-\rho_0)$ may be recycled toward the creation of $M'^{(b+1)}_1$, and an amount $(1-\gamma)(\rho_2-\rho_0)$ may be recycled toward the creation of $M^{(b+1)}_2$.
The key observations here are that (i) this procedure ensures that, with high probability, \textit{both} Encoder~1 and Encoder~2 know messages $M_0^{(b)}, M'^{(b)}_1, M''^{(b)}_1$ at the end of block~$b$, so that they can coordinate their choices of $M_0^{(b+1)}$; and (ii) the dependencies across blocks are only created through $M''^{(b)}_1$, which is nearly independent of $Z_b^r$.

Before formalizing the reasoning above, note that the effective rate of \textit{new} randomness for Encoder~1 in block $b$ is
\begin{align}
R_1 \triangleq \rho_1+\rho_2-\gamma(\rho_2-\rho_0)=\rho_1+(1-\gamma)\rho_2+\gamma \rho_0,
\end{align}
and the effective rate for Encoder~2 is
\begin{align}
R_2 \triangleq \rho_3-(1-\gamma)(\rho_2-\rho_0).
\end{align}
Using the values of $\rho_0,\rho_1,\rho_2,\rho_3$ chosen in~(\ref{rho3})-(\ref{rho4}), we may obtain all\footnote{For $R_1$ we choose $\gamma=1$ and for $R_2$ we choose $\gamma=0$, in each case finding the smallest single-user rate constraint so that the entire rate region is captured. The sum-rate constraint is independent of $\gamma$.}  rate pairs such that
\begin{align*}
R_1 &\geq \rho_1+\rho_0 = I(UX_1;Z)+2\epsilon,\\
R_2 &\geq \rho_3 - \rho_2+\rho_0 = I(X_1X_2;Z) - H(X_1|U) +4\epsilon,\\
R_1+R_2 &\geq \rho_1+\rho_2+\rho_3-(\rho_2-\rho_0)= I(X_1X_2;Z) +3\epsilon,
\end{align*} 
which is the desired rate region.

\begin{figure*}
\centering
\includegraphics[width=0.9\textwidth]{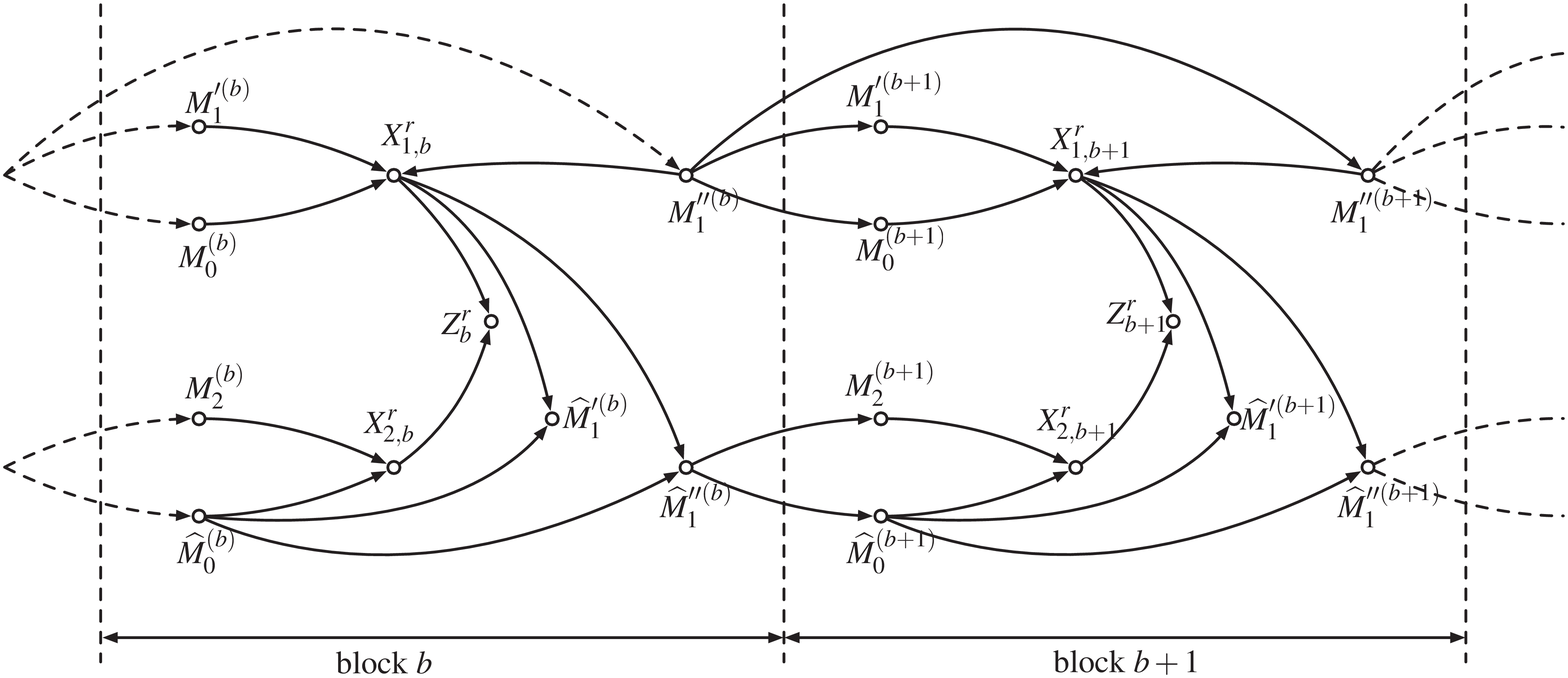}
\caption{Functional dependence graph for the block-Markov encoding scheme for the channel resolvability of the \ac{MAC} with strictly-causal cribbing.}
\label{fig:fdg}
\end{figure*}

It remains to show that this coding scheme guarantees that $\mathbb{D}(P_{Z^n}||Q_Z^{\otimes n})$ can be bounded following the expansion in~\eqref{eq6}.
First, note that for every $b \in \llbracket 2,B \rrbracket$,
\begin{align}
I(Z_b^r;Z_{b+1}^{B,r}) &\leq I(Z_b^r;M''^{(b)}_1\widehat{M}''^{(b)}_1 Z_{b+1}^{B,r})\\
&= I(Z_b^r;M''^{(b)}_1\widehat{M}''^{(b)}_1 ) \label{d1}
\end{align}
since the Markov chain $Z_b^r \rightarrow M''^{(b)}_1\widehat{M}''^{(b)}_1 \rightarrow Z_{b+1}^{B,r}$ holds as seen in the functional dependence graph shown in Fig.~\ref{fig:fdg}. Next, note that
\begin{align}
&I (Z_b^r ; M''^{(b)}_1\widehat{M}''^{(b)}_1 ) \twocolbreak
 \hspaceonetwocol{0.0in}{0.2in} \leq I(Z_b^r ; M''^{(b)}_1 ) +H(\widehat{M}''^{(b)}_1 | M''^{(b)}_1) \\
& \hspaceonetwocol{1.23in}{0.2in} \overset{(a)}{\leq} \mathbb{D}(P_{Z_b^r M''^{(b)}_1 }||P_{Z_b^r} P_{M''^{(b)}_1 })+H(P_e^{(b)}) + P_e^{(b)} r \rho_2 \\
& \hspaceonetwocol{1.23in}{0.2in} \overset{(b)}{\leq} \mathbb{D}(P_{Z_b^r M''^{(b)}_1 }||Q_Z^{\otimes r} P_{M''^{(b)}_1 })+H(P_e^{(b)}) + P_e^{(b)} r \rho_2 
\end{align}
where $(a)$ follows from Fano's inequality and $(b)$ follows from 
\begin{align}
&\mathbb{D}(P_{Z_b^r M''^{(b)}_1 }||P_{Z_b^r} P_{M''^{(b)}_1 })\twocolbreak
 = \mathbb{D}(P_{Z_b^r M''^{(b)}_1 }||Q_Z^{\otimes r} P_{M''^{(b)}_1 })  -\mathbb{D}(P_{Z_b^r}||Q_Z^{\otimes r}). \label{d2}
\end{align}

Combining (\ref{d1})-(\ref{d2}) into (\ref{eq6}), we finally obtain
\begin{align}
\mathbb{D}(P_{Z^n}||Q_Z^{\otimes n}) \leq & 2 \sum_{b=1}^{B}\mathbb{D}(P_{Z_b^r M''^{(b)}_1 }||Q_Z^{\otimes r} P_{M''^{(b)}_1 }) \twocolbreak
+B\left(H(P_e^{(b)}) + P_e^{(b)} r \rho_2 \right). \label{v4}
\end{align}
Finally, we show that $\mathbb{D}(P_{Z_b^r M''^{(b)}_1 }||Q_Z^{\otimes r} P_{M''^{(b)}_1 })$ is not too different from $\mathbb{D}(\bar{P}_{Z_b^r M''^{(b)}_1 }||Q_Z^{\otimes r} \bar{P}_{M''^{(b)}_1 })=\mathbb{D}^{(b)}$, where we recall that $\bar{P}$ is induced when both encoders use $M_0^{(b)}$, while $P$ is induced when Encoder~1 uses $M_0^{(b)}$ and Encoder~2 uses an estimate $\widehat{M}_0^{(b)}$ derived from his estimate $\widehat{M}''^{(b-1)}_1$. The total variation $\mathbb{V}(P_{Z_b^r M''^{(b)}_1},\bar{P}_{Z_b^r M''^{(b)}_1})$ satisfies
\begin{align}
&\mathbb{V}(P_{Z_b^r M''^{(b)}_1},\bar{P}_{Z_b^r M''^{(b)}_1}) \twocolbreak 
 \hspaceonetwocol{0.0in}{0.0in}\leq \mathbb{V}(P_{Z_b^r M_0^{(b)}M'^{(b)}_1M''^{(b)}_1\widehat{M}_0^{(b)}M_2^{(b)}},\bar{P}_{Z_b^r M_0^{(b)}M'^{(b)}_1M''^{(b)}_1{M}_0^{(b)}M_2^{(b)}})\\
&\hspaceonetwocol{1.5in}{0.0in}=\mathbb{V}(P_{M_0^{(b)}\widehat{M}_0^{(b)}},\bar{P}_{M_0^{(b)}{M}_0^{(b)}})\\
&\hspaceonetwocol{1.5in}{0.0in} = 2\P{M_0^{(b)}\neq\widehat{M}_0^{(b)}}\\
& \hspaceonetwocol{1.5in}{0.0in}\leq 2\mathbb{P}(M''^{(b-1)}_1 \neq \widehat{M}''^{(b-1)}_1).
\end{align}
Consequently, since $\bar{P}_{M''^{(b-1)}_1} = P_{M''^{(b-1)}_1}$, we obtain
\begin{align}
\mathbb{V}(P_{Z_b^r M''^{(b)}_1}, Q_Z^{\otimes r} P_{M''^{(b)}_1}) &\leq \mathbb{V}( P_{Z_b^r M''^{(b)}_1},\bar{P}_{Z_b^r M''^{(b)}_1 })\twocolbreak
+ \mathbb{V}(\bar{P}_{Z_b^r M''^{(b)}_1 },Q_Z^{\otimes r} P_{M''^{(b)}_1}) \nonumber\\
&\leq 2 \times 2^{-\alpha' r} + 2^{-\frac{\beta'}{2} r}. \label{v1}
\end{align}
where we have used Pinsker's inequality to bound the last term. To conclude that $\mathbb{D}(P_{Z_b^r M''^{(b)}_1} || Q_Z^{\otimes r} P_{M''^{(b)}_1})$ vanishes, we recall the following result~\cite[Eq. (323)]{Sason2016a}.

\begin{lemma}\label{lemma1}
Let $P$ and $Q$ be two distributions on a finite alphabet $\mathcal{A}$ such that $P$ is absolutely continuous \ac{wrt} $Q$. If $\mu \triangleq \min_{a\in\calQ:Q(a)>0} Q(a)$, we have
\begin{align*}
\mathbb{D}(P||Q) \leq \log\left(\frac{1}{\mu}\right)\mathbb{V}(P,Q).
\end{align*}
\end{lemma}
Note that $P_{Z_b^r M''^{(b)}_1}$ is absolutely continuous \ac{wrt} $Q_Z^{\otimes r} P_{M''^{(b)}_1}$ by definition of $Q_Z$ and the code construction. Hence, using (\ref{v1}) together with Lemma~\ref{lemma1} shows that there exists $\eta>0$ such that for all $r$ large enough
\begin{align}
\mathbb{D}(P_{Z_b^r M''^{(b)}_1} || Q_Z^{\otimes r} P_{M''^{(b)}_1}) < 2^{-\eta r}. \label{v2}
\end{align}
Substituting (\ref{v2}) and (\ref{v3}) into (\ref{v4}) shows the desired result.

\begin{remark}
\label{Remark:noncooperative}
 Recall that some private common randomness is required to jump-start the block-Markov encoding; this common randomness can be collected during a non-cooperative starting phase in the following manner. The two encoders will start transmitting with rates $R_1=H(X_1)$ and $R_2=0$, which exceeds the single-user resolvability rate. Simultaneously, via the usual arguments in the degraded wiretap channel $M_1 \rightarrow X_1^n \rightarrow Z^n$, one can convey $\frac{1}{n}(I(M_1;X_1^n)-I(M_1;Z^n)) = H(X_1|Z)$ bits of randomness from User~1 to User~2 while keeping it independent of $Z$ and maintaining an i.i.d. distribution $Q^{\otimes n}(z)$. By collecting this randomness for $\frac{\rho_0}{H(X_1|Z)}$ blocks, sufficient common randomness will be available to start the block-Markov process. The difference of rates $(R_1, R_2)$ in the starting phase will be amortized over $B$ blocks, with $B$ growing without bound, thus the average rates remain as described. The concept of starting the block-Markov transmission with a non-cooperative phase goes back to the inception of block-Markov encoding~\cite{Cover_ElGamal_Relay}.
\end{remark}
\begin{remark}
The above mentioned initialization of block-Markov coding leaves open the possibility that some $Q(z)$ may be compatible with some joint distribution $p(x_1,x_2)$ but incompatible with all product distributions $p(x_1)p(x_2)$. Such a $Q(z)$ is valid for cooperative transmission but cannot be generated during the non-cooperative initialization of block-Markov encoding. Thus, for a more precise definition of the model for MAC with {\em strictly causal} cribbing, in the context of resolvability, we are presented with three distinct choices: Either (a)~some private shared randomness (with rate that amortizes asymptotically to zero) is made available to the model, or (b) the distribution~$Q(z)$ is limited to the set that can be generated via product distributions $p(x_1)p(x_2)$, or (c)~the distribution of $Z^n$, although still i.i.d., is allowed to deviate from the target $Q(z)$ for a finite number of blocks at the beginning of transmission. Options (b) and (c) are both reasonable for secrecy applications of resolvability; option (b) may affect secrecy rates.
\end{remark}


\subsubsection{Converse} \hfil\\
We consider a $(2^{nR_1},2^{nR_2},n)$ code such that $\mathbb{D}(P_{Z^n}||Q_Z^{\otimes n}) \leq \epsilon$. Then,

\begin{align}
n  R_1 &= H(M_1) \nonumber\\
&\geq I(M_1;Z^n) \nonumber\\
&\stackrel{(a)}{=} I(M_1,X_1^n;Z^n) \label{bbb}\\
&\geq I(X_1^n;Z^n) \nonumber\\
&= I(X_1^n,X_2^n;Z^n) - I(X_2^n;Z^n|X_1^n) \nonumber\\
&\stackrel{(b)}{\geq} \sum_{x_1^n} \sum_{x_2^n} \sum_{z^n} P(x_1^n,x_2^n,z^n) \log \frac{W^{\otimes n}(z^n|x_1^n,x_2^n)}{P_{Z^n}(z^n)} \twocolbreak
\hspaceonetwocol{0.0in}{0.2in}- \sum_i I(X_{2i};Z_i|X_{1i},X_1^{i-1}) \label{eq18}\\
&\stackrel{(c)}{=} \sum_{x_1^n} \sum_{x_2^n} \sum_{z^n} P(x_1^n,x_2^n,z^n) \log \frac{W^{\otimes n}(z^n|x_1^n,x_2^n)}{Q_Z^{\otimes n}(z^n)} \twocolbreak
 \hspaceonetwocol{0.0in}{0.2in}-\mathbb{D}(P_{Z^n} || Q_Z^{\otimes n})- \sum_i I(X_{2i};Z_i|X_{1i},U_i) \label{eq19}\\
&\geq \sum_i \sum_{u_i} \sum_{x_{1i}} \sum_{x_{2i}} \sum_{z_i} P(u_i,x_{1i},x_{2i},z_i) \log \frac{W(z_i|x_{1i},x_{2i})}{Q(z_i)} \nonumber\\
&\hspaceonetwocol{1.5in}{0.2in}  - \! \sum_i \! \sum_{u_i} \! \sum_{x_{1i}} \! \sum_{x_{2i}} \! \sum_{z_i} P(u_i,x_{1i},x_{2i},z_i) \log \frac{W(z_i|x_{1i},x_{2i})}{P(z_i|x_{1i},u_i)} \twocolbreak
 \hspaceonetwocol{0.0in}{0.2in}-\epsilon \label{aaa} \\
&= \sum_i \sum_{u_i} \sum_{x_{1i}} \sum_{x_{2i}} \sum_{z_i} P(u_i,x_{1i},x_{2i},z_i) \log \frac{P(z_i|x_{1i},u_i)}{Q(z_i)} \twocolbreak
 \hspaceonetwocol{0.0in}{0.2in}-\epsilon \nonumber\\
&= \sum_i \sum_{u_i} \sum_{x_{1i}}  \sum_{z_i} P(u_i,x_{1i},z_i) \log \frac{P(z_i|x_{1i},u_i)}{Q(z_i)} -\epsilon \nonumber\\
&= \sum_i \mathbb{D}(P_{U_i,X_{1i},Z_i}||P_{U_i,X_{1i}} Q_{Z_i})-\epsilon \nonumber\\
&=  \sum_i I(U_i X_{1i};Z_i) + \sum_i \mathbb{D}(P_{Z_i}||Q_Z)-\epsilon\\
&\stackrel{(d)}{\geq} n I(U_Q X_{1Q};Z_Q|Q) -\epsilon  \label{eqc1}\\
& = nI(Q U_Q X_{1Q};Z_Q) -nI(Q;Z_Q)-\epsilon\\
&\stackrel{(e)}{\geq} n I(U X_1;Z) -n\epsilon ' \label{eqc2}
\end{align}
where 
\begin{enumerate}[(a)]
    \item follows from the definition of the deterministic encoding functions in~(\ref{encfn_strict});
    \item follows from $I(X_2^n;Z^n|X_1^n)\leq \sum_{i=1}^nI(X_{2,i};Z_i|X_{1,i}X_1^{i-1})$ since the channel is memoryless  and because conditioning does not increase entropy;
    \item follows setting $U_i\triangleq X_1^{i-1}$;
    \item  follows by introducing a random variable $Q$ uniformly distributed on $\llbracket 1,n \rrbracket$ and independent of all others;
    \item follows by \cite[Lemma VI.3]{synthesis} for some $\epsilon'>0$ with $\lim_{\epsilon \rightarrow 0}\epsilon' =0$ and by setting $U=(Q,U_Q)$, $X_1=X_{1Q}$ and $Z=Z_Q$.
\end{enumerate}
Notice that upon setting $X_2=X_{2Q}$ and recalling that the cribbing is strictly-causal such that $X_{2Q} $ is a function of $(M_2,Q,U_Q)$, and that the Markov chains $M_1,X_1 \to U \to M_2$ and $X_{1Q} \to Q,U_Q \to X_{2Q}$ hold, we have
\begin{align}
P_{Q U_Q X_{1Q} X_{2Q}Z_Q}&=P_{Q U_Q}P_{X_{1Q}|U_Q  Q}P_{X_{2Q}|U_Q Q}W_{Z_Q|X_{1Q} X_{2Q}},\\
&\text{and}  \nonumber\\
P_{UX_1X_2Z}&=P_UP_{X_1|U}P_{X_2|U}W_{Z|X_1X_2}.
\end{align} 

Next, note that 
\begin{align}
n R_2 &= H(M_2) \nonumber\\
&\geq H(M_2|X_1^n) \nonumber\\
&\geq I(M_2;Z^n|X^n_1) \nonumber\\
&\stackrel{(a)}{=} I(M_2,X_2^n;Z^n|X_1^n) \label{ccc}\\
&\geq I(X_2^n;Z^n|X_1^n) \nonumber\\
&= I(X_1^n,X_2^n;Z^n) - I(X_1^n;Z^n) \label{eqc8} \\
&= \sum_{x_1^n} \sum_{x_2^n} \sum_{z^n} P(x_1^n,x_2^n,z^n) \log \frac{W^{\otimes n}(z^n|x_1^n,x_2^n)}{ P_{Z^n}(z^n)}\twocolbreak
 \hspaceonetwocol{0.0in}{0.2in}- I(X_1^n;Z^n) \nonumber\\
&\geq \sum_{x_1^n} \sum_{x_2^n} \sum_{z^n} P(x_1^n,x_2^n,z^n) \log \frac{W^{\otimes n}(z^n|x_1^n,x_2^n)}{ Q_Z^{\otimes n}(z^n)} \twocolbreak
\hspaceonetwocol{0.0in}{0.2in}-\mathbb{D}(P_{Z^n}||Q_Z^{\otimes n})- H(X_1^n) \nonumber\\
&\geq \sum_i \sum_{x_{1i}} \sum_{x_{2i}} \sum_{z_i} P(x_{1i},x_{2i},z_i) \log \frac{W(z_i|x_{1i},x_{2i})}{Q(z_i)} \twocolbreak
\hspaceonetwocol{0.0in}{0.2in}-\sum_i H(X_{1i}|U_i) -\epsilon \nonumber\\
&= \sum_i \sum_i \sum_{x_{1i}} \sum_{x_{2i}} P(x_{1i},x_{2i},z_i) \log \frac{W(z_i|x_{1i},x_{2i})}{P(z_i)} \twocolbreak
\hspaceonetwocol{0.0in}{0.2in}+\mathbb{D}(P_{Z_i}||Q_Z) -\sum_i H(X_{1i}|U_i) -\epsilon \nonumber\\
&\geq \sum_i I(X_{1i},X_{2i};Z_i)-\sum_i H(X_{1i}|U_i) -\epsilon \nonumber\\
&= nI(X_{1Q}X_{2Q};Z_Q|Q)-nH(X_{1Q}|U_Q Q) -\epsilon \nonumber\\
&=nI(Q X_{1Q}X_{2Q};Z_Q)-nI(Q;Z_Q)\twocolbreak
-nH(X_{1Q}|U_Q Q) -\epsilon \nonumber\\
&\stackrel{(b)}{\geq} nI(X_{1Q}X_{2Q};Z_Q)-nH(X_{1Q}|U_Q Q) -n\epsilon ' \nonumber\\
&= nI(X_1X_2;Z) -nH(X_1|U)-n\epsilon '   \label{eqc7}
\end{align}
where $(a)$ follows from the definition of the encoding function and $(b)$ follows by \cite[Lemma VI.3]{synthesis} for some $\epsilon'>0$ with $\lim_{\epsilon \rightarrow 0}\epsilon' =0$.
 Finally,
\begin{align}
n (R_1+R_2)
&=    H(M_1,M_2)  \nonumber\\
&\geq I(M_1,M_2;Z^n) \nonumber\\
&\geq I(X_1^n,X_2^n;Z^n) +\mathbb{D}(P_{Z^n} ||Q_Z^{\otimes n}) - \epsilon \nonumber\\
&\geq nI(X_1,X_2;Z) - n \epsilon ' \nonumber
\end{align}
where we have merely repeated the steps in (\ref{eqc8})-(\ref{eqc7}). To conclude the converse proof, we again refer the reader to~\cite[Section VI.C]{synthesis} for the continuity of the resolvability region at $\epsilon \to 0$.


\subsection{Theorem~\ref{th_causal}: \ac{MAC} with Causal Cribbing} \label{causal}

\subsubsection{Achievability}\hfil\\
The proof of the \ac{MAC} with causal cribbing is similar to the \ac{MAC} with strictly-causal cribbing, however, we use a Shannon strategy to generate $X_2$ rather than codewords~\cite{cribbingmeulen}.
Let $\mathscr{T} \eqdef\mathcal{X}_2^{|\mathcal{X}_1|}$ be the set of all strategies that map $\mathcal{X}_1$ into $\mathcal{X}_2$, and for $t\in\mathcal{T}$ denote by $t(x_1)\in\calX_2$ the image of $x_1\in\calX_1$. The \ac{MAC} induced by the Shannon strategy is denoted by $(\mathcal{X}_1 \times \mathscr{T}, W^{+}_{Z|X_1,T}, \mathcal{Z} )$ where $W^{+}_{Z|X_1,T} \triangleq W_{Z|X_1,X_2=T(X_1)} $.

By Theorem~\ref{th_strict}, rate pairs $(R_1,R_2)$ satisfying the following conditions are achievable with strictly-causal cribbing.
\begin{align}
R_1 &> I(U,X_1;Z),\\
R_2&> I(X_1,T;Z)-H(X_1|U),\\
R_1+R_2&> I(X_1,T;Z).
\end{align}
with $H(X_1|U) > I(U,X_1;Z)$ for any joint distribution $P_{UX_1TZ}\eqdef P_{U}P_{X_1|U}P_{T|U}W^{+}_{Z|X_1T}$ with marginal $Q_Z$. Restricting the distribution to satisfy $P_{UX_1TZ}\eqdef P_{U} P_{X_1}P_{T}W^{+}_{Z|X_1T}$ yields:
\begin{align*}
H(X_1|U)&=H(X_1), \\
I(U,X_1;Z)&=I(X_1;Z), \\
I(X_1,T;Z)&=I(X_1,X_2,T;Z)=I(X_1,X_2;Z),
\end{align*}
and $P(x_1,x_2,z)=P(x_1)\sum_{t:t(x_1)=x_2}P(t)W(z|x_1,x_2)$. This is possible because of the fact that for an arbitrary distribution $P^{*}(x_1,x_2)$, there always exists a product distribution $P(x_1,t)=P(x_1)P(t)$ such that $P^*(x_1,x_2)=P(x_1) \sum_{t:t(x_1)=x_2}P(t)$.
This is achieved by choosing~\cite[Eq.~(44)]{cribbingmeulen}
\begin{align*}
    &P(x_1)=\sum_{x_2}P^*(x_1,x_2),\\
    &P(t)=\prod_{x_1} \frac{P^*(x_1,x_2=t(x_1))}{P(x_1)}.
\end{align*}
Note that the constraint $H(X_1|U)>I(U,X_1;Z)$ is now automatically satisfied if $H(X_1|Z)>0$.

The achievability scheme presented thus far depends on $H(X_1|Z)>0$. The same achievable rates can be attained for $H(X_1|Z)=0$, however, a different scheme is required for this extremal case, which is presented below.

Consider a distribution $P(x_1,x_2)=P(x_1)P(x_2|x_1)$ such that $\sum_{x_1,x_2} P(x_1,x_2)W(z|x_1,x_2)=Q_Z$. 
\begin{itemize}
\item Independently generate $2^{nR_1}$ codewords $x_1^n$ each with probability $P(x_1^n)$. Label them $x_1^n(m_1)$, $m_1 \in \llbracket 1,2^{nR_1}\rrbracket$.
\item For every $x_1^n(m_1)$, independently generate $2^{nR_2}$ codewords $x_2^n$ each with probability $P(x_2^n|x_1^n(m_1))=P_{X_2|X_1(m_1)}^{\otimes n}(x_2^n|x_1^n(m_1))$. Label them $x_{2}^n(x_1^n(m_1),m_2)$,  $m_{2} \in \llbracket 1,2^{nR_2}\rrbracket$.
\end{itemize}

For $m_{1} \in \llbracket 1,2^{nR_1} \rrbracket$ and $m_{2} \in \llbracket 1,2^{nR_2}\rrbracket$, let $X_1^n(m_1)$ and $X^n_2(X_1^n(m_1),m_2)$ denote  random variables representing the randomly generated codewords. The average \ac{KL} divergence can be bounded as:

\begin{align}
&\mathbb{E} \big( \mathbb{D} ( P_{Z^n}||Q_{Z}^{\otimes n})\big)  \twocolbreak
\hspaceonetwocol{0.0in}{0.2in} =\mathbb{E} \bigg(\sum_{z^n} \sum_{x_1^n} P(x_1^n,z^n) \log \frac{\sum_{x_1^n} P(x_1^n,z^n)}{\sum_{x_1^n} Q_{X_1,Z}^{\otimes n}(x_1^n,z^n)}\bigg) \nonumber\\
& \hspaceonetwocol{1.2in}{0.2in}\stackrel{(a)}{\leq} \mathbb{E} \bigg( \sum_{z^n} \sum_{x_1^n} P(x_1^n,z^n) \log \frac{ P(x_1^n,z^n)}{ Q_{X_1,Z}^{\otimes n}(x_1^n,z^n)} \bigg) \nonumber\\
&\hspaceonetwocol{1.2in}{0.2in}= \mathbb{E} \big( \mathbb{D} ( P_{X_1^n,Z^n}||Q_{X_1,Z}^{\otimes n})\big)  \nonumber\\
&\hspaceonetwocol{1.2in}{0.2in}\stackrel{(b)}{=} \mathbb{E} \big( \mathbb{D} ( P_{Z^n|X_1^n}||Q_{Z|X_1}^{\otimes n}|P_{X_1^n})\big)+\mathbb{E} \big( \mathbb{D} ( P_{X_1^n}||P_{X_1}^{\otimes n} ) \big) \label{eq:eff}
\end{align}
where\\
$Q_{X_1,Z}=\sum_{x_2}P(x_1,x_2)W(z|x_1,x_2)$;
\begin{enumerate}[(a)]
    \item follows by the log-sum inequality;
    \item follows from the chain rule of \ac{KL} divergence.
\end{enumerate}

\begin{figure*}
\centering
\includegraphics[width=0.7\textwidth]{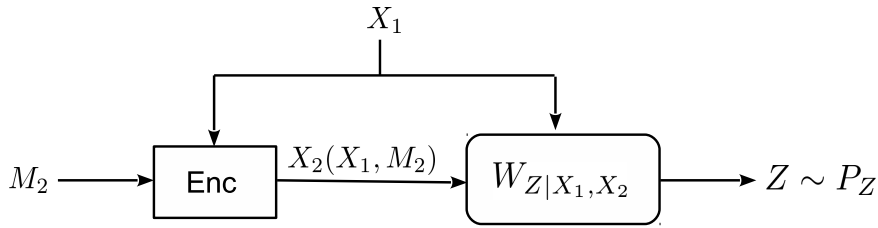} 
\caption{State-dependent point-to-point channel.}
\label{fig:fig_state_dep}
\end{figure*}

Let $R_1>H(X_1)$, in which case the second term of \eqref{eq:eff} vanishes as $n\to\infty$ and the channel is effectively a state-dependent point-to-point channel (Fig.~\ref{fig:fig_state_dep}). Using similar bounding techniques as those used earlier in this paper (e.g. the proof of \eqref{rho00}-\eqref{rho03}), the first term of \eqref{eq:eff} vanishes as $n\to\infty$ if $R_2>I(X_2;Z|X_1)$. Since $H(X_1|Z)=0$, the achievable region is
\begin{align*}
    R_1>H(X_1)&=I(X_1;Z),\\
    R_2>I(X_2;Z|X_1)&=I(X_1,X_2;Z)-H(X_1),\\
    R_1+R_2>H(X_1)+I(X_2;Z|X_1)&=I(X_1,X_2;Z).
\end{align*}

\subsubsection{Converse}\hfil\\
Since the \ac{MAC} with causal cribbing is a special case of the \ac{MAC} with non-causal cribbing, it follows that the converse of the latter holds for the causal cribbing scenario.


\section{Conclusion}
\label{sec:conc}

This paper develops inner and outer bounds for the resolvability rates of the multiple-access channel with cribbing that is either non-causal, causal, strictly causal, or in the form of degraded message sets. The derived inner and outer bounds are tight except for the strictly causal case. The key insight of this paper and the corresponding analytical contribution involves the  hiding of the cooperation among the encoders so that the dependencies created by cooperation are undetectable at the channel output.  This is made possible because the cribbing mechanism, even with a strict-causality constraint, creates an effective wiretap channel and allows the exchange of secret information. In the context of channel resolvability, this secret information plays the role of randomness that can be reused for cooperation without impacting the desired output approximation. Then, deriving secrecy from channel resolvability, achievable strong secrecy rates were derived for \ac{MAC} wiretap channels with cribbing.


\appendices

\section{Proof of Convexity of Resolvability Regions}
\label{appendix:convexity}

\subsection{Convexity of \ac{MAC} with Degraded Message Sets}\label{appendix:convexity_degraded}

Assume that $(R_1^{(1)},R_2^{(1)})$ and $(R_1^{(2)},R_2^{(2)})$ are achievable, which implies the existence of two distributions $P_{X_1,X_2,Z}^{(1)}$ and $P_{X_1,X_2,Z}^{(2)}$ with marginal $Q_Z$ such that,
\begin{align*}
R_1^{(1)} &\geq I(X_1^{(1)};Z^{(1)}),          \\    R_1^{(1)}+R_2^{(1)} &\geq I(X_1^{(1)},X_2^{(1)};Z^{(1)}),       \\
\text{and} \\
R_1^{(2)} &\geq I(X_1^{(2)};Z^{(2)}), \\
R_1^{(2)}+R_2^{(2)} &\geq I(X_1^{(2)},X_2^{(2)};Z^{(2)}).
\end{align*}

Let $P_{X_1,X_2|Z}^{(3)}=\lambda P_{X_1,X_2|Z}^{(1)} +(1-\lambda)P_{X_1,X_2|Z}^{(2)}$ for $\lambda \in \llbracket 0,1 \rrbracket$ and  $P_{X_1|Z}^{(3)}=\lambda P_{X_1|Z}^{(1)} +(1-\lambda)P_{X_1|Z}^{(2)}$. 
Note that $P^{(3)}_{X_1,X_2,Z}$ resulting from a convex combination of $P^{(1)}_{X_1,X_2,Z}$ and $P^{(2)}_{X_1,X_2,Z}$ exists unlike the \ac{MAC} with non-cooperating encoders, where the convex combination does not necessarily factorize into a product distribution.

From the convexity of $I(X_1,X_2;Z)$ with respect to $P_{X_1,X_2|Z}$ and the convexity of $I(X_1;Z)$ with respect to $P_{X_1|Z}$, it follows that for a fixed $Q_Z$:
\begin{align*}
I(X_1^{(3)};Z^{(3)}) &\leq \lambda I(X_1^{(1)};Z^{(1)})  +(1-\lambda) I(X_1^{(2)};Z^{(2)}),  \\
I(X_1^{(3)},X_2^{(3)};Z^{(3)}) &\leq \lambda I(X_1^{(1)},X_2^{(1)};Z^{(1)})  +(1-\lambda) I(X_1^{(2)},X_2^{(2)};Z^{(2)}).
\end{align*} 
Therefore we have 
\begin{align*}
I(X_1^{(3)};Z^{(3)}) &\leq \lambda R_1^{(1)} +(1-\lambda) R_1^{(2)},  \\
I(X_1^{(3)},X_2^{(3)};Z^{(3)}) &\leq \lambda (R_1^{(1)}+R_2^{(1)} )  +(1-\lambda) (R_1^{(2)}+R_2^{(2)} ).
\end{align*} 
which implies that $\big(\lambda R_1^{(1)} +(1-\lambda) R_1^{(2)},\lambda R_2^{(1)} +(1-\lambda) R_2^{(2)} \big)$ is inside the achievable region defined by $P^{(3)}_{X_1,X_2,Z}$.

\subsection{Convexity of \ac{MAC} with Non-Causal/Causal Cribbing}\label{appendix:convexity_causal}
Similar to the \ac{MAC} with degraded message sets, assume that $(R_1^{(1)},R_2^{(1)})$ and $(R_1^{(2)},R_2^{(2)})$ are achievable, which implies the existence of two distributions $P_{X_1,X_2,Z}^{(1)}$ and $P_{X_1,X_2,Z}^{(2)}$ with marginal $Q_Z$ such that,
\begin{align*}
R_1^{(1)} &\geq I(X_1^{(1)};Z^{(1)}),          \\   
R_2^{(1)} & \geq   I(X_1^{(1)},X_2^{(1)};Z^{(1)}) - H(X_1^{(1)}), \\                      
R_1^{(1)}+R_2^{(1)} &\geq I(X_1^{(1)},X_2^{(1)};Z^{(1)}),       \\
&\text{and} \\
R_1^{(2)} &\geq I(X_1^{(2)};Z^{(2)}), \\
R_2^{(1)} & \geq   I(X_1^{(2)},X_2^{(2)};Z^{(2)}) - H(X_1^{(2)}), \\  
R_1^{(2)}+R_2^{(2)} &\geq I(X_1^{(2)},X_2^{(2)};Z^{(2)}).
\end{align*}

Let $P_{X_1,X_2|Z}^{(3)}=\lambda P_{X_1,X_2|Z}^{(1)} +(1-\lambda)P_{X_1,X_2|Z}^{(2)}$ for $\lambda \in \llbracket 0,1 \rrbracket$. Then   $P_{X_1|Z}^{(3)}=\lambda P_{X_1|Z}^{(1)} +(1-\lambda)P_{X_1|Z}^{(2)}$ and $P_{X_1}^{(3)}=\lambda P_{X_1}^{(1)} +(1-\lambda)P_{X_1}^{(2)} $.

From the convexity of $I(X_1,X_2;Z)$ with respect to $P_{X_1,X_2|Z}$, the convexity of $I(X_1;Z)$ with respect to $P_{X_1|Z}$ for a fixed $Q_Z$ and the concavity of $H(X_1)$ with respect to $P_{X_1}$:
\begin{align*}
I(X_1^{(3)};Z^{(3)}) &\leq \lambda I(X_1^{(1)};Z^{(1)})  +(1-\lambda) I(X_1^{(2)};Z^{(2)}),  \\
H(X_1^{(3)}) &\geq  \lambda H(X_1^{(1)})  + (1-\lambda) H(X_1^{(2)}), \\
I(X_1^{(3)},X_2^{(3)};Z^{(3)}) &\leq \lambda I(X_1^{(1)},X_2^{(1)};Z^{(1)})  +(1-\lambda) I(X_1^{(2)},X_2^{(2)};Z^{(2)}).
\end{align*}

Therefore we have 
\begin{align*}
I(X_1^{(3)};Z^{(3)}) &\leq \lambda R_1^{(1)} +(1-\lambda) R_1^{(2)},  \\
I(X_1^{(3)},X_2^{(3)};Z^{(3)})  - H(X_1^{(3)}) &\leq \lambda R_2^{(1)} +(1-\lambda) R_2^{(2)},  \\
I(X_1^{(3)},X_2^{(3)};Z^{(3)}) &\leq \lambda (R_1^{(1)}+R_2^{(1)} )  +(1-\lambda) (R_1^{(2)}+R_2^{(2)} ).
\end{align*} 
which implies that $\big(\lambda R_1^{(1)} +(1-\lambda) R_1^{(2)},\lambda R_2^{(1)} +(1-\lambda) R_2^{(2)} \big)$ is inside the achievable region defined by $P^{(3)}_{X_1,X_2,Z}$.

\subsection{Convexity of \ac{MAC} with Strictly-Causal Cribbing}\label{appendix:convexity_strict}

To prove the convexity of the inner bound, assume that $(R_1^{(1)},R_2^{(1)})$ and $(R_1^{(2)},R_2^{(2)})$ are achievable, which implies the existence of two distributions $P_{U,X_1,X_2,Z}^{(1)}=P_{U}^{(1)}P_{X_1|U}^{(1)}P_{X_2|U}^{(1)}W_{Z|X_1,X_2}$ and $P_{U,X_1,X_2,Z}^{(2)}=P_{U}^{(2)}P_{X_1|U}^{(2)}P_{X_2|U}^{(2)}W_{Z|X_1,X_2}$ with marginal $Q_Z$ such that,
\begin{align*}
R_1^{(1)} &\geq I(U^{(1)},X_1^{(1)};Z^{(1)}),          \\   
R_2^{(1)} & \geq   I(X_1^{(1)},X_2^{(1)};Z^{(1)}) - H(X_1^{(1)}|U^{(1)}), \\                      
R_1^{(1)}+R_2^{(1)} &\geq I(X_1^{(1)},X_2^{(1)};Z^{(1)}),       \\
\text{with  } & H(X_1^{(1)}|U^{(1)})>I(U^{(1)},X_1^{(1)};Z^{(1)}),\\
&\text{and} \\
R_1^{(2)} &\geq I(U^{(2)},X_1^{(2)};Z^{(2)}), \\
R_2^{(1)} & \geq   I(X_1^{(2)},X_2^{(2)};Z^{(2)}) - H(X_1^{(2)}|U^{(2)}), \\  
R_1^{(2)}+R_2^{(2)} &\geq I(X_1^{(2)},X_2^{(2)};Z^{(2)}),\\
\text{with  } & H(X_1^{(2)}|U^{(2)})>I(U^{(2)},X_1^{(2)};Z^{(2)}).
\end{align*}

For $\lambda\in \llbracket 0,1 \rrbracket$, let $Q\in\{1,2\}$ with $\text{Pr}(Q=1)=\lambda$ and $\text{Pr}(Q=2)=1-\lambda$. Define $U^{(3)}\triangleq(U^{(Q)},Q)$, $X_1^{(3)}\triangleq X_1^{(Q)}$, $X_2^{(3)}\triangleq X_2^{(Q)}$ and $Z^{(3)}\triangleq Z^{(Q)}$. Let $Q$, $(U^{(1)},X_1^{(1)},X_2^{(1)},Z^{(1)})$ and $(U^{(2)},X_1^{(2)},X_2^{(2)},Z^{(2)})$ be independent so that $P_{U,X_1,X_2,Z}^{(3)}=\lambda P_{U,X_1,X_2,Z}^{(1)}+(1-\lambda)P_{U,X_1,X_2,Z}^{(2)}$ can be written as $P_{U,X_1,X_2,Z}^{(3)}=P_{U}^{(3)}P_{X_1|U}^{(3)}P_{X_2|U}^{(3)}W_{Z|X_1,X_2}$. From the definition of the random variables:
\begin{align*}
    H(X_1^{(3)}|U^{(3)})=\lambda H(X_1^{(1)}|U^{(1)}) + (1-\lambda)H(X_1^{(2)}|U^{(2)}).
\end{align*}

For a fixed $Q_Z$, we have $P_{X_1,X_2|Z}^{(3)}=\lambda P_{X_1,X_2|Z}^{(1)} +(1-\lambda)P_{X_1,X_2|Z}^{(2)}$ and $P_{U,X_1|Z}^{(3)}=\lambda P_{U,X_1|Z}^{(1)} +(1-\lambda)P_{U,X_1|Z}^{(2)}$. From the convexity of $I(U,X_1;Z)$ with respect to $P_{U,X_1|Z}$ and the convexity of $I(X_1,X_2;Z)$ with respect to $P_{X_1,X_2|Z}$:
\begin{align*}
I(U^{(3)},X_1^{(3)};Z^{(3)}) &\leq \lambda I(U^{(1)},X_1^{(1)};Z^{(1)})  +(1-\lambda) I(U^{(2)}X_1^{(2)};Z^{(2)}),  \\
I(X_1^{(3)},X_2^{(3)};Z^{(3)}) &\leq \lambda I(X_1^{(1)},X_2^{(1)};Z^{(1)})  +(1-\lambda) I(X_1^{(2)},X_2^{(2)};Z^{(2)}).
\end{align*}

Therefore, we have 
\begin{align*}
I(U^{(3)},X_1^{(3)};Z^{(3)}) &\leq \lambda R_1^{(1)} +(1-\lambda) R_1^{(2)},  \\
I(X_1^{(3)},X_2^{(3)};Z^{(3)})  - H(X_1^{(3)}|U^{(3)}) &\leq \lambda R_2^{(1)} +(1-\lambda) R_2^{(2)},  \\
I(X_1^{(3)},X_2^{(3)};Z^{(3)}) &\leq \lambda (R_1^{(1)}+R_2^{(1)} )  +(1-\lambda) (R_1^{(2)}+R_2^{(2)} ).
\end{align*} 
and
\begin{align*}
H(X_1^{(3)}|U^{(3)}) &=\lambda H(X_1^{(1)}|U^{(1)}) + (1-\lambda)H(X_1^{(2)}|U^{(2)})\\
& > \lambda I(U^{(1)},X_1^{(1)};Z^{(1)})+(1-\lambda)I(U^{(2)}X_1^{(2)};Z^{(2)})\\
& \geq I(U^{(3)},X_1^{(3)};Z^{(3)}).
\end{align*}
which implies that $\big(\lambda R_1^{(1)} +(1-\lambda) R_1^{(2)},\lambda R_2^{(1)} +(1-\lambda) R_2^{(2)} \big)$ is inside the achievable region defined by $P^{(3)}_{U,X_1,X_2,Z}$. The convexity of the outer bound is proven similarly but without the entropy constraint $H(X_1|U)>I(U,X_1;Z)$.  

\section{Proof of Equations \eqref{rho00}-\eqref{rho03} }
\label{appendix:resolv_rates}

For $m_0 \in \llbracket 1,2^{n\rho_0}\rrbracket$, $m_1' \in \llbracket 1,2^{n\rho_1}\rrbracket$ and $m_{2} \in \llbracket 1,2^{n\rho_3}\rrbracket$, let $U^r(m_0)$, $X_1(m_0,m_1',m_1'')$ and $X_2(m_0,m_2)$  denote the random variables representing the randomly generated codewords.

\begin{align}
& \mathbb{E} (   \mathbb{D}(\bar{P}_{Z_b^r,M_1''^{(b)}}||Q_Z^{\otimes r} \bar{P}_{M_1''^{(b)}} )    ) \nonumber\\
& = \mathbb{E} \sum_{m_1'',z_b^r} \bar{P}(m_1'',z_b^r) \log \frac{\bar{P}(m_1'',z_b^r)}{\bar{P}_{M_1''^{(b)}}(m_1'') Q_Z^{\otimes r}(z^r)} \nonumber\\
& = \mathbb{E} \sum_{m_1''} \bar{P}(m_1'') \sum_{z_b^r} \bar{P}(z_b^r|m_1'') \log \frac{\bar{P}(z_b^r|m_1'')}{Q_Z^{\otimes r}(z^r)} \nonumber\\
& = \mathbb{E}  \sum_{m_1''} \bar{P}(m_1'') \sum_{z_b^r} \sum_{i,j,k} \frac{ W^{\otimes r}(z_b^r|U^r(i),X_1^r(i,j,m_1''),X_2^r(i,k))}{2^{r(\rho_0+\rho_1+\rho_3)}} \nonumber\\
&\hspaceonetwocol{1.5in}{0.2in}  \log \sum_{i',j',k'}\frac{  W^{\otimes r}(z_b^r|U^r(i'),X_1^r(i',j',m_1''),X_2^r(i',k'))}{2^{r(\rho_0+\rho_1+\rho_3)}Q_Z^{\otimes r}(z^r)} \nonumber\\
& \stackrel{(a)}{\leq}\frac{1}{2^{r(\rho_0+\rho_1+\rho_3)}} \sum_{m_1''} \bar{P}(m_1'') \sum_{i,j,k} \sum_{z_b^r} \sum_{u^r(i)} \sum_{x_1^r(i,j,m_1'')} \sum_{x_2^r(i,k)} \twocolbreak
\hspaceonetwocol{0in}{0.2in}  \bar{P}(u^r(i),x_1^r(i,j,m_1''),x_2^r(i,k),z^r) \nonumber\\
&\hspaceonetwocol{1.5in}{0.2in}   \log\mathbb{E}_{\setminus (i,j,k)}  \!\! \sum_{i',j',k'} \!\!  \frac{  W^{\otimes r}(z_b^r|U^r(i'),X_1^r(i',j',m_1''),X_2^r(i',k'))}{2^{r(\rho_0+\rho_1+\rho_3)}Q_Z^{\otimes r}(z^r)} \nonumber\\
& \leq \frac{1}{2^{r(\rho_0+\rho_1+\rho_3)}} \sum_{m_1''} \bar{P}(m_1'') \sum_{i,j,k} \sum_{z_b^r} \sum_{u^r(i)} \sum_{x_1^r(i,j,m_1'')} \sum_{x_2^r(i,k)} \twocolbreak
\hspaceonetwocol{0in}{0.2in}  \bar{P}(u^r(i),x_1^r(i,j,m_1''),x_2^r(i,k),z^r) \nonumber\\
&\hspaceonetwocol{1.5in}{0.2in}  \log \mathbb{E}_{\setminus (i,j,k)} \frac{1}{2^{r(\rho_0+\rho_1+\rho_3)}Q_Z^{\otimes r}(z^r)} \twocolbreak
\hspaceonetwocol{0in}{0.2in}  \bigg( W^{\otimes r} (z^r|u^r(i),x_1^r(i,j,m_1''),x_2^r(i,k)) \nonumber\\
&\hspaceonetwocol{1.5in}{0.2in}  + \sum_{\substack{j' \neq j\\ k' \neq k }} W^{\otimes r} (z^r|u^r(i),X_1^r(i,j',m_1''),X_2^r(i,k'))  \nonumber\\
&\hspaceonetwocol{1.5in}{0.2in} + \sum_{k' \neq k} W^{\otimes r} (z^r|u^r(i),x_1^r(i,j,m_1''),X_2^r(i,k'))  \nonumber\\
&\hspaceonetwocol{1.5in}{0.2in}  + \sum_{j' \neq j} W^{\otimes r} (z^r|u^r(i),X_1^r(i,j',m_1''),x_2^r(i,k))  \nonumber\\
&\hspaceonetwocol{1.5in}{0.2in}  + \sum_{\substack{i' \neq i\\j',k'}} W^{\otimes r} (z^r|U^r(i'),X_1^r(i',j',m_1''),X_2^r(i',k'))  \bigg) \nonumber\\
& \leq \frac{1}{2^{r(\rho_0+\rho_1+\rho_3)}} \sum_{m_1''} \bar{P}(m_1'') \sum_{i,j,k} \sum_{z_b^r} \sum_{u^r(i)} \sum_{x_1^r(i,j,m_1'')} \sum_{x_2^r(i,k)} \twocolbreak
\hspaceonetwocol{0in}{0.2in}  \bar{P}(u^r(i),x_1^r(i,j,m_1''),x_2^r(i,k),z^r) \nonumber\\
&\hspaceonetwocol{1.5in}{0.2in}   \log  \frac{1}{2^{r(\rho_0+\rho_1+\rho_3)}Q_Z^{\otimes r}(z^r)}  \twocolbreak
\hspaceonetwocol{0in}{0.2in} \bigg( W^{\otimes r} (z^r|u^r(i),x_1^r(i,j,m_1''),x_2^r(i,k)) \nonumber\\
&\hspaceonetwocol{1.5in}{0.2in}  + \sum_{\substack{j' \neq j\\ k' \neq k }} \bar{P}^{\otimes r} (z^r|u^r(i))   + \sum_{k' \neq k} \bar{P}^{\otimes r} (z^r|u^r(i),x_1^r(i,j,m_1''))  \nonumber\\
&\hspaceonetwocol{1.5in}{0.2in}  + \sum_{j' \neq j}
\bar{P}^{\otimes r} (z^r|u^r(i),x_2^r(i,k))  +1 \bigg)  \nonumber\\
& =\Psi_1+\Psi_2 \nonumber\\
\end{align}
where
\begin{enumerate}[(a)]
\item follows by Jensen's inequality.
\end{enumerate}
\begin{align}
\Psi_1 &\triangleq  \frac{1}{2^{r(\rho_0+\rho_1+\rho_3)}} \sum_{m_1''} \bar{P}(m_1'') \sum_{i,j,k}  \twocolbreak
\hspaceonetwocol{0in}{0.2in} \sum_{(u^r(i),x_1^r(i,j,m_1''),x_2^r(i,k),z_b^r)\in \mathcal{T}_\epsilon^r(P_{U,X_1,X_2,Z})} \nonumber\\
& \hspaceonetwocol{0.2in}{0.2in}\bar{P}(u^r(i),x_1^r(i,j,m_1''),x_2^r(i,k),z^r)  \twocolbreak
\hspaceonetwocol{0in}{0.2in} \log  \frac{1}{2^{r(\rho_0+\rho_1+\rho_3)}Q_Z^{\otimes r}(z^r)}  \nonumber\\
& \hspaceonetwocol{0.2in}{0.2in}\bigg( W^{\otimes r} (z^r|u^r(i),x_1^r(i,j,m_1''),x_2^r(i,k)) \twocolbreak
 \hspaceonetwocol{0in}{0.2in}+ \sum_{\substack{j' \neq j\\ k' \neq k }}   \bar{P}^{\otimes r} (z^r|u^r(i)) \nonumber\\
& \hspaceonetwocol{0.2in}{0.2in}+ \sum_{k' \neq k} \bar{P}^{\otimes r} (z^r|u^r(i),x_1^r(i,j,m_1''))  \twocolbreak 
 \hspaceonetwocol{0in}{0.2in}+ \sum_{j' \neq j} \bar{P}^{\otimes r} (z^r|u^r(i),x_2^r(i,k))  +1 \bigg)  \nonumber\\
&  \leq \log \Big( \frac{2^{-r(1-\epsilon)H(Z|X_1,X_2)}}{2^{r(\rho_0+\rho_1+\rho_3)}2^{-r(1+\epsilon)H(Z)}} \twocolbreak\hspaceonetwocol{0in}{0.2in}
+ \frac{2^{-r(1-\epsilon)H(Z|U)}}{2^{r\rho_0}2^{-r(1+\epsilon)H(Z)}} 
  +\frac{2^{-r(1-\epsilon)H(Z|U,X_1)}}{2^{r(\rho_0+\rho_1)}2^{-r(1+\epsilon)H(Z)}} \nonumber\\
&\hspaceonetwocol{3.5in}{0.2in}+\frac{2^{-r(1-\epsilon)H(Z|U,X_2)}}{2^{r(\rho_0+\rho_3)}2^{-r(1+\epsilon)H(Z)}} +1 \Big)\nonumber\\
 &\leq \log \Big( 2^{-r(\rho_0+\rho_1+\rho_3-I(X_1,X_2;Z)-2\epsilon H(Z))} \twocolbreak
\hspaceonetwocol{0in}{0.2in}+ 2^{-r(\rho_0-I(U;Z)-2\epsilon H(Z))}\twocolbreak
\hspaceonetwocol{0in}{0.2in}+ 2^{-r(\rho_0+\rho_1-I(U,X_1;Z)-2\epsilon H(Z))}  \nonumber\\
&\hspaceonetwocol{3.5in}{0.2in}+ 2^{-r(\rho_0+\rho_3-I(U,X_2;Z)-2\epsilon H(Z))} +1 \Big) \nonumber
\end{align}

\begin{align}
&\Psi_2 \triangleq \sum_{m_1''} \bar{P}(m_1'') \sum_i \sum_j \sum_k  \twocolbreak
\hspaceonetwocol{0in}{0.2in}\sum_{(u^r(i),x_1^r(i,j,m_1''),x_2^r(i,k),z_b^r) \notin \mathcal{T}_\epsilon^r(P_{U,X_1,X_2,Z})} \nonumber\\
&\hspaceonetwocol{0.2in}{0.2in}\bar{P}(u^r(i),x_1^r(i,j,m_1''),x_2^r(i,k),z^r) \twocolbreak
\hspaceonetwocol{0in}{0.2in} \log  \frac{1}{2^{r(\rho_0+\rho_1+\rho_3)}Q_Z^{\otimes r}(z^r)}  \nonumber\\
&\hspaceonetwocol{0.2in}{0.2in}\bigg( W^{\otimes r} (z^r|u^r(i),x_1^r(i,j,m_1''),x_2^r(i,k)) \twocolbreak
\hspaceonetwocol{0in}{0.2in} + \sum_{\substack{j' \neq j\\ k' \neq k }}  \bar{P}^{\otimes r} (z^r|u^r(i)) \nonumber\\
&\hspaceonetwocol{0.2in}{0.2in} + \sum_{k' \neq k}  \bar{P}^{\otimes r} (z^r|u^r(i),x_1^r(i,j,m_1''))  + \sum_{j' \neq j} \bar{P}^{\otimes r} (z^r|u^r(i),x_2^r(i,k))  +1 \bigg)  \nonumber \\
& \leq 2|\mathcal{U}| |\mathcal{X}_1| |\mathcal{X}_2| |\mathcal{Z}| e^{-r \epsilon^2 \mu_{U X_1 X_2 Z}} r\log (\frac{4}{\mu_Z}+1) \nonumber
\end{align}
 where
 \begin{align}
\mu_Z &=  \min_{z \in \mathcal{Z}} Q(z) \nonumber\\
\mu_{UX_1X_2Z} &= \min_{(u,x_1,x_2,z) \in (\mathcal{U},\mathcal{X}_1,\mathcal{X}_2,\mathcal{Z}) } Q(u,x_1,x_2,z) \nonumber
\end{align}

Combining the bounds on $\Psi_1$ and $\Psi_2$, $\mathbb{E} (   \mathbb{D}(\bar{P}_{Z_b^r,M_1''^{(b)}}||Q_Z^{\otimes r} \bar{P}_{M_1''^{(b)}} )    )\xrightarrow[]{r\to \infty} 0$  when \eqref{rho00}-\eqref{rho03} are satisfied.

\section{Achievability Proofs of the Strong Secrecy Region for \ac{MAC} with Cribbing}

\subsection{Achievability: Strong Secrecy of \ac{MAC} with Degraded Message Sets }
\label{appendix:mac_degraded}

Consider a distribution $P(x_1,x_2)=P(x_1)P(x_2|x_1)$ such that $\sum_{x_1,x_2} P(x_1,x_2) W(z|x_1,x_2)=Q_Z(z)$.

\textbf{Code Construction:}
\begin{itemize}
\item Independently generate $2^{n(R_1+R'_1)}$ codewords $x_1^n$ each with probability $P(X_1^n)=P_{X_1}^{\otimes n}(x_1^n)$. Label them $x_1^n(m_1,m'_1)$, $m_1 \in \llbracket 1,2^{nR_1} \rrbracket$ and $m'_1 \in \llbracket 1,2^{nR'_1}\rrbracket$.
\item For every $x_1^n(m_1,m_1')$, independently generate $2^{n(R_2+R'_2)}$ codewords $x_2^n$ each with probability $P(x_2^n|x_1^n(m_1,m_1'))=P_{X_2|X_1}^{\otimes n}(x_2^n|x_1^n(m_1,m_1'))$. Label them $x_{2}^n(m_1,m_1',m_2,m_2')$,  $m_2 \in \llbracket 1,2^{nR_2}\rrbracket$ and $m'_2 \in \llbracket 1,2^{nR'_2}\rrbracket$.
\end{itemize}

\textbf{Encoding:}
 To send $m_1$, Encoder~1 transmits $x_1^n(m_1,m'_1)$. To send $m_2$, Encoder~2 cooperatively sends $x_2^n(m_1,m'_1,m_2,m'_2)$. $m'_1$ and $m'_2$ are independently chosen
at random from $\llbracket 1,2^{nR'_1}\rrbracket$ and $\llbracket 1,2^{nR'_2}\rrbracket$ respectively.

\textbf{Decoding:}
 The decoder finds $(m_1,m'_1,m_2,m'_2)$ such that \[(x_1^n(m_1,m'_1), x_2^n(m_1,m'_1,m_2,m'_2), y^n) \in \mathcal{T}_{\epsilon}^{(n)}(P_{X_1,X_2,Y}).
 \]
 
 \textbf{Probability of error analysis:}
 Using standard arguments, the probability of error averaged over all codebooks vanishes exponentially with $n$ if
\begin{align}
R_2+R'_2 &< I(X_2;Y|X_1) \label{eq70}\\
R_1+R'_1+R_2+R'_2 &< I(X_1,X_2;Y) \label{eq71}
\end{align} 

\textbf{Secrecy analysis:}
We will show that the information leakage, averaged over all codebooks vanishes, exponentially with $n$. We use the results of Theorem~\ref{th_degraded} to bound $\mathbb{E}_{M_1,M_2}[\mathbb{D}(P_{Z^n|M_1,M_2}||Q_Z^{\otimes n})]$ such that the channel output distribution at the wiretapper is, on average, independent of the transmitted messages and follows the i.i.d distribution $Q_Z^{\otimes n}$. This is sufficient to ensure secrecy because $I(M_1,M_2;Z^n)$ can be bounded by $\mathbb{E}_{M_1,M_2}[\mathbb{D}(P_{Z^n|M_1,M_2}||Q_Z^{\otimes n})]$, as follows:

\begin{align}
& I(M_1,M_2;Z^n) \twocolbreak
\hspaceonetwocol{0in}{0.2in} =\mathbb{D}(P_{M_1,M_2,Z^n}||P_{M_1,M_2}P_{Z^n}) \\
& \hspaceonetwocol{1in}{0.2in} = \sum_{m_1,m_2,z^n} P_{M_1,M_2,Z^n} (m_1,m_2,z^n) \twocolbreak
\hspaceonetwocol{0in}{0.4in}\log \frac{P_{M_1,M_2,Z^n} (m_1,m_2,z^n) }{P_{M_1,M_2}(m_1,m_2)P_{Z^n}(z^n)}\\
& \hspaceonetwocol{1 in}{0.2in} = \sum_{m_1,m_2} P_{M_1,M_2}(m_1,m_2) \mathbb{D}( P_{Z^n|M_1,M_2} || P_{Z^n} ) \label{eqqa}\\
& \hspaceonetwocol{1in}{0.2in} \stackrel{(a)}{\leq} \mathbb{E}_{M_1,M_2} \left( \mathbb{D} (P_{Z^n|M_1,M_2} || Q_Z^{\otimes n}) \right) \;, \label{eqa}
\end{align}
where (a) follows by adding $\mathbb{D}(P_{Z^n}||Q_Z^{\otimes n}) \geq 0$ to \eqref{eqqa}. With $P_{Z|M_1M_2}(z|m_1,m_2)=2^{-n(R_1'+R_2')}\sum_{i,j} W^{\otimes n}(z^n|x_1(m_1,i),x_2(m_1,i,m_2,j))$ and applying Theorem~\ref{th_degraded} to (\ref{eqa}), $I(M_1,M_2;Z^n) $ vanishes exponentially with $n$ if

\begin{align}
R'_1 &> I(X_1;Z) \label{eq76}\\
R'_1+R'_2 &> I(X_1,X_2;Z) \label{eq77}
\end{align}

Combining \eqref{eq70}, \eqref{eq71}, \eqref{eq76} and \eqref{eq77}, and using Fourier-Motzkin elimination, the following rate region is achievable
\begin{align}
    R_2&<I(X_2;Z|X_1), \\
    R_1+R_2&<I(X_1,X_2;Z).
\end{align}


\subsection{Achievability: Strong Secrecy of \ac{MAC} with Non-Causal Cribbing }
\label{appendix:mac_noncausal}

Consider a distribution  $P(x_1,x_2)=P(x_1)P(x_2|x_1)$ such that $\sum_{x_1,x_2} P(x_1,x_2) W(z|x_1,x_2)=Q_Z(z)$.

\textbf{Code Construction:}
\begin{itemize}
\item Independently generate $2^{n(R_1+R'_1)}$ codewords $x_1^n$ each with probability $P(X_1^n)=P_{X_1}^{\otimes n}(x_1^n)$. Label them $x_1^n(m_1,m'_1)$, $m_1 \in \llbracket 1,2^{nR_1}\rrbracket$ and $m'_1 \in \llbracket 1,2^{nR'_1}\rrbracket$.
\item For every $x_1^n(m_1,m'_1)$, independently generate $2^{n(R_2+R'_2)}$ codewords $x_2^n$ each with probability $P(x_2^n|x_1^n(m_1,m'_1))=P_{X_2|X_1}^{\otimes n}(x_2^n|x_1^n(m_1,m'_1))$. Label them $x_{2}^n(x_1^n(m_1,m'_1),m_2,m'_2)$,  $m_2 \in\llbracket 1,2^{nR_2}\rrbracket$ and $m'_2 \in \llbracket 1,2^{nR'_2}\rrbracket$.
\end{itemize}

We assume that each message is chosen independently and uniformly from its corresponding set.
As a result of cribbing, Encoder~2 knows $x_1^n$ in advance, therefore before transmission, it finds $(\hat{m}_1,\hat{m}_1')$ such that $(x_1^n(\hat{m}_1,\hat{m}_1'),x_1^n) \in \mathcal{T}_{\epsilon}^{(n)}(P_{X_1,X_1})$ where $(\hat{m}_1,\hat{m}_1')$ are the estimates of $(m_1,m_1')$.

\textbf{Encoding:}
To send $m_1$, Encoder~1 sends $x_1^n(m_1,m_1')$. To send $m_2$, Encoder~2 cooperatively sends $x_2^n(x_1^n(\hat{m}_1,\hat{m}_1'),m_2,m_2')$. 

\textbf{Decoding at the receiver:}
 The decoder finds $(\hat{\hat{m}}_1,\hat{\hat{m}}'_1,\hat{\hat{m}}_2,\hat{\hat{m}}'_2)$ such that $(x_1^n(\hat{\hat{m}}_1,\hat{\hat{m}}_1'),x_2^n(x_1^n(\hat{\hat{m}}_1,\hat{\hat{m}}_1'),\hat{\hat{m}}_2,\hat{\hat{m}}_2'),y^n) \in \mathcal{T}_{\epsilon}^{(n)}(P_{X_1,X_2,Y})$.

\textbf{Probability of error analysis and secrecy analysis:}
We follow similar steps as the cognitive \ac{MAC} case. Using standard arguments to bound the probability of error and using the resolvability results of Theorem~\ref{th_noncausal} we get
\begin{align}
R_1+R'_1 &< H(X_1) \label{eqw}\\
R_2+R'_2 &< I(X_2;Y|X_1)\\
R_1+R'_1+R_2+R'_2 &< I(X_1,X_2;Y)\\
R'_1 &> I(X_1;Z)\\
R'_2 &> I(X_1,X_2;Z)-H(X_1) \label{eqx}\\
R'_1+R'_2 &> I(X_1,X_2;Z) \label{eqy}
\end{align} 
Note that (\ref{eqx}) is implied by (\ref{eqw}) and (\ref{eqy}). Using Fourier-Motzkin elimination, the following rate region is achievable
\begin{align}
R_1 &< H(X_1) -I(X_1;Z), \\
R_2 &< I(X_2;Y|X_1),\\
R_1+R_2 &< I(X_1,X_2;Y)-I(X_1,X_2;Z). 
\end{align} 


\subsection{Achievability: Strong Secrecy of \ac{MAC} with Strictly-Causal Cribbing }
\label{appendix:mac_strictlycausal}

\begin{figure*}
\centering
\includegraphics[width=0.9\textwidth]{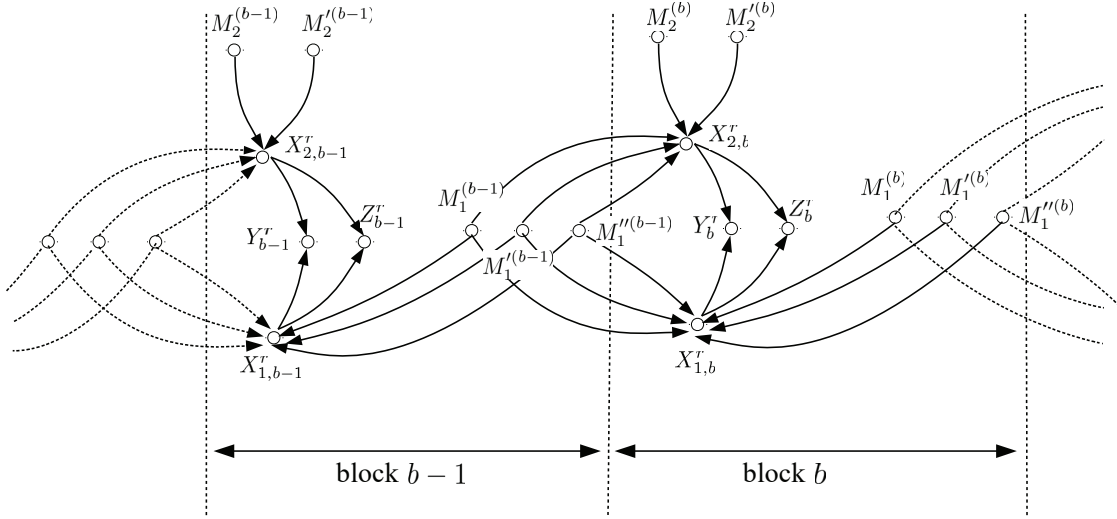} 
\caption{Functional dependence graph for the block-Markov encoding scheme for the \ac{MAC} with strictly-causal cribbing.}
\label{fig:fdg_secrecy}
\end{figure*}

The interesting challenge being addressed in this section is that the decoding at output $Y$ and secrecy at output $Z$ have dissonant requirements under strictly causal cribbing. For decoding, Willems and van der Muelen~\cite{cribbingmeulen} proposed a block-Markov superposition coding technique where all the information carried by the cribbing signal is used as cloud centers for the purpose of cooperation between the two encoders. The tightness of the inner and outer bounds in~\cite{cribbingmeulen}  strongly suggests (via continuity arguments in joint probability distributions) that leaving out any part of cribbing information from cooperation can incur a rate loss for decoding at $Y$. On the other hand, the resolvability results of this paper strongly suggest that for simulating a desired probability distribution at $Z$, it is beneficial to have a local randomness component at $X_1$ that does not take part in cooperation with $X_2$. The contribution of this section is to produce a coding strategy that reconciles these two dissonant requirements.

We begin by  informally describing the main idea of this section with a simplified notation. The codebook for $X_1$ is driven by three variables: $(M_1,M_1',M_1'')$. $M_1$ is the secret message, and $M_1',M_1''$ are uniformly distributed dithers. The encoder for $X_2$ will decode the cribbing signal $X_1$ and use all its three components as cloud center for the next transmission, which we call $M_0,M_0',M_0''$. This, as mentioned earlier, is crucial for decoding at $Y$. Now we introduce an additional constraint (enforced by proper assignment of rates) so that $Z$ is independent of $M_0'$, one of the dither components of $X_1$. Thus, as far as the distribution of $Z$ is concerned, one of the two dither components of $X_1$ is local (private) to $X_1$ and is not used by $X_2$. To elaborate further, due to the imposed independence, the cloud centers that only differ in their $M_0'$ index must give rise to the same distribution in $Z$, therefore Encoder 2 has been effectively enjoined from cooperation or coordination with part of the dither of Encoder~1, which for all practical purposes becomes local to $X_1$ as far as the eavesdropper is concerned. We now make these ideas precise in full detail, which includes direct reference to block indices as well as accounting for discrepancies between message/dither indices and their estimated values.

We use a combination of block-Markov encoding and backward decoding. Independently and uniformly distributed messages $m_{1}^{(b)} \in \llbracket 1,2^{rR_1}\rrbracket$ and $m_{2}^{(b)} \in \llbracket 1,2^{rR_2}\rrbracket$  will be sent over $B$ blocks. Each block consists of $r$ transmissions so that $n=rB$. 
Consider a distribution $P(u,x_1,x_2)=P(u)P(x_1|u)P(x_2|u)$ such that $\sum_{u,x_1,x_2} P(u,x_1,x_2) W(z|x_1,x_2)=Q_Z(z)$. 

\textbf{Code Construction:}
In each block $b \in \llbracket 1,B \rrbracket$:
\begin{itemize}
\item Independently generate $2^{r(R_1+\rho_1'+\rho_1'')}$ codewords $u_b^r$ each with probability $P(u^r)=P_{U}^{\otimes r}(u^r)$. Label them $u^r(m_{0}^{(b)},m_{0}'^{(b)},m_{0}''^{(b)})$, $m_{0}^{(b)} \in \llbracket 1,2^{rR_1}\rrbracket$, $m_{0}'^{(b)} \in \llbracket 1,2^{r\rho_1'}\rrbracket$ and $m_{0}''^{(b)} \in \llbracket 1,2^{r\rho_1''}\rrbracket$.
\item For every $u^r(m_{0}^{(b)},m_{0}'^{(b)},m_{0}''^{(b)})$, independently generate $2^{r(R_1+\rho_1'+\rho_1'')}$ codewords $x_{1b}^r$ each with probability $P(x_1^r|u^r(m_{0}^{(b)},m_{0}'^{(b)},m_{0}''^{(b)}))=P_{X_1|U}^{\otimes r}(x_1^r|u^r(m_{0}^{(b)},m_{0}'^{(b)},m_{0}''^{(b)}))$. Label them $x_{1}^r(m_{0}^{(b)},m_{0}'^{(b)},m_{0}''^{(b)},m_{1}^{(b)},m_{1}'^{(b)},m_{1}''^{(b)})$,  $m_{1}^{(b)} \in \llbracket 1,2^{rR_1}\rrbracket$, $m_{1}'^{(b)} \in \llbracket 1,2^{r\rho_1'}\rrbracket$ and $m_{1}''^{(b)} \in \llbracket 1,2^{r\rho_1''}\rrbracket$.
\item For every $u^r(m_{0}^{(b)},m_{0}'^{(b)},m_{0}''^{(b)})$, independently generate 
$2^{r(R_2+\rho_2)}$ codewords $x_{2b}^r$ each with probability $P(x_2^r|u^r(m_{0}^{(b)},m_{0}'^{(b)},m_{0}''^{(b)}))=P_{X_2|U}^{\otimes r}(x_2^r|u^r(m_{0}^{(b)},m_{0}'^{(b)},m_{0}''^{(b)}))$.
Label them $x_{2}^r(m_{0}^{(b)},m_{0}'^{(b)},m_{0}''^{(b)},m_{2}^{(b)},m_{2}'^{(b)})$,  $m_{2}^{(b)} \in \llbracket 1,2^{rR_2}\rrbracket$ and $m_{2}'^{(b)} \in \llbracket 1,2^{r\rho_2}\rrbracket$.
\end{itemize}

We intend to use these codebooks in the following manner:
\begin{enumerate}
    \item Block Markov encoding via $M_0^{(b)}=M_1^{(b-1)}$, $M_0'^{(b)}=M_1'^{(b-1)}$ and $M_0''^{(b)}=M_1''^{(b-1)}$; 
    \item $M_1^{(b)}$, $M_1'^{(b)}$ and $M_1''^{(b)}$ can be decoded from $X_{1b}^r$ knowing $(M_0^{(b)},M_0'^{(b)},M_0''^{(b)})$;
    \item \{$M_1^{(1)},\dots,M_1^{(B)}\}$ and $\{M_2^{(1)},\dots,M_2^{(B)}\}$ are secret from $\{Z_1^r,\dots,Z_B^r$\};
    \item $M_1''^{(b)}$ is the common randomness to be used by both encoders in block $b+1$;
    \item $M_1'^{(b)}$ is local randomness used by Encoder~1 and $M_2'^{(b)}$ is local randomness used by Encoder~2;
    \item The messages $M_0^{(b)}$, $M_0'^{(b)}$, $M_0''^{(b)}$, $M_1^{(b)}$, $M_1'^{(b)}$, $M_1''^{(b)}$, $M_2^{(b)}$ and $M_2'^{(b)}$ can be decoded at the receiver from $Y_b^r$ and the messages decoded in future blocks $b+1$ to $B$ (backward decoding).
\end{enumerate}

As a result of cribbing, after block $b$, Encoder~2 finds estimates $(\hat{m}_1^{(b)},\hat{m}_1'^{(b)},\hat{m}_1''^{(b)} )$ for $(m_1^{(b)},m_1'^{(b)},m_1''^{(b)} )$ such that 
\begin{align}
(u^r(\hat{m}_0^{(b)},\hat{m}_0'^{(b)},\hat{m}_0''^{(b)}),x_{1}^r(\hat{m}_0^{(b)},\hat{m}_0'^{(b)},\hat{m}_0''^{(b)},\hat{m}_1^{(b)},\hat{m}_1'^{(b)},\hat{m}_1''^{(b)}),x_{1b}^r) \linesplit
\in \mathcal{T}_{\epsilon}^{(r)}(P_{U,X_1,X_1}).
\end{align} 
where $(\hat{m}_0^{(b)},\hat{m}_0'^{(b)},\hat{m}_0''^{(b)})=(\hat{m}_1^{(b-1)},\hat{m}_1'^{(b-1)},\hat{m}_1''^{(b-1)})$.

\textbf{Encoding:}
We apply block-Markov encoding as follows. 
In block $b$, the encoders send:
\begin{align*}
 x_{1b}^r &= x_1^m(m_0^{(b)},m_0'^{(b)},m_0''^{(b)},m_1^{(b)},m_1'^{(b)},m_1''^{(b)}) \\
 x_{2b}^r &= x_2^m(\hat{m}_0^{(b)},\hat{m}_0'^{(b)},\hat{m}_0''^{(b)},m_2^{(b)},m_2'^{(b)})
\end{align*}
where $(m_0^{(b)},m_0'^{(b)},m_0''^{(b)})=(m_1^{(b-1)},m_1'^{(b-1)},m_1''^{(b-1)})$ and $(\hat{m}_0^{(b)},\hat{m}_0'^{(b)},\hat{m}_0''^{(b)})=(\hat{m}_1^{(b-1)},\hat{m}_1'^{(b-1)},\hat{m}_1''^{(b-1)})$. We also assume that the encoders and decoder have access to $(M_0^{(1)},M_0'^{(1)},M_0''^{(1)},M_1^{(B)},M_1'^{(B)},M_1''^{(B)},M_2^{(B)},M_2'^{(B)})$ through private common randomness.

 \textbf{Decoding at the receiver:} 
 The legitimate receiver waits until all $B$ blocks are transmitted and then performs backward decoding. The decoder first finds $(\hat{\hat{m}}_0^{(B)},\hat{\hat{m}}_0'^{(B)},\hat{\hat{m}}_0''^{(B)})$ such that 
\begin{align*}
(u^r(\hat{\hat{m}}_0^{(B)},\hat{\hat{m}}_0'^{(B)},\hat{\hat{m}}_0''^{(B)}),x^r_1(\hat{\hat{m}}_0^{(B)},\hat{\hat{m}}_0'^{(B)},\hat{\hat{m}}_0''^{(B)},\hat{\hat{m}}_1^{(B)},\hat{\hat{m}}_1'^{(B)},\hat{\hat{m}}_1''^{(B)}),\nonumber\\
x^r_2(\hat{\hat{m}}_0^{(B)},\hat{\hat{m}}_0'^{(B)},\hat{\hat{m}}_0''^{(B)},\hat{\hat{m}}_2^{(B)},\hat{\hat{m}}_2'^{(B)}),
 y^r_B)\in \mathcal{T}_{\epsilon}^{(r)}(P_{U,X_1,X_2,Y}).
\end{align*}
Assuming that $(m_0^{(B)},m_0'^{(B)},m_0''^{(B)})$, $(m_0^{(B-1)},m_0'^{(B-1)},m_0''^{(B-1)})$, \ldots, $(m_0^{(b+1)},m_0'^{(b+1)},m_0''^{(b+1)})$ have been decoded, the decoder sets $(\hat{\hat{m}}_1^{(b)},\hat{\hat{m}}_1'^{(b)},\hat{\hat{m}}_1''^{(b)})=(\hat{\hat{m}}_0^{(b+1)},\hat{\hat{m}}_0'^{(b+1)},\hat{\hat{m}}_0''^{(b+1)})$ and finds $(\hat{\hat{m}}_0^{(b)},\hat{\hat{m}}_0'^{(b)},\hat{\hat{m}}_0''^{(b)})$ and $(\hat{\hat{m}}_2^{(b)},\hat{\hat{m}}_2'^{(b)})$ such that
\begin{align*}
(u^r(\hat{\hat{m}}_0^{(b)},\hat{\hat{m}}_0'^{(b)},\hat{\hat{m}}_0''^{(b)}),x^r_1(\hat{\hat{m}}_0^{(b)},\hat{\hat{m}}_0'^{(b)},\hat{\hat{m}}_0''^{(b)},\hat{\hat{m}}_1^{(b)},\hat{\hat{m}}_1'^{(b)},\hat{\hat{m}}_1''^{(b)}),\nonumber\\
x^r_2(\hat{\hat{m}}_0^{(b)},\hat{\hat{m}}_0'^{(b)},\hat{\hat{m}}_0''^{(b)},\hat{\hat{m}}_2^{(b)},\hat{\hat{m}}_2'^{(b)}),y^r_b) \in \mathcal{T}_{\epsilon}^{(r)}(P_{U,X_1,X_2,Y}).
\end{align*}

 \textbf{Probability of error analysis:}
  Using the arguments for error analysis from~\cite[Lemma 4]{permuter2}, the probability of error of each block vanishes exponentially with $r$ and in turn vanishes across blocks  if 
     \begin{align}
 R_1+\rho_1'+\rho_1'' &< H(X_1|U), \label{mika11}\\
 R_2+\rho_2 &< I(X_2;Y|X_1,U),\\
 R_1+\rho_1'+\rho_1''+R_2+\rho_2&< I(X_1,X_2;Y). \label{mika22}
 \end{align}

\textbf{Secrecy analysis:}
Let $\bar{P}$ be the probability induced when  both encoders use $(M_0^{(b)},M_0'^{(b)},M_0''^{(b)})$. Let $P$ be the probability when Encoder~1 uses $(M_0^{(b)},M_0'^{(b)},M_0''^{(b)})$ and Encoder~2 uses the estimate $(\hat{M}_0^{(b)},\hat{M}_0'^{(b)},\hat{M}_0''^{(b)})$. For the secrecy analysis we find conditions so that $I(M_0^{(b)},M_1^{(b)},M_1'^{(b)},M_1''^{(b)},\hat{M}_0^{(b)},\hat{M}_0'^{(b)},M_2^{(b)};Z_b^r)$ vanishes exponentially with $r$. This is motivated by: 
\begin{itemize}
\item
$(M_1^{(b)}, M_0^{(b)}, \hat{M}_0^{(b)}, M_2^{(b)})$ are the Encoder~1 secret message in the present and the past, the estimate of the latter (at Encoder~2), and Encoder~2 secret message, which must be kept secret from $Z_b^r$, obviously.
\item
$(M_1^{(b)},M_1'^{(b)},M_1''^{(b)})$ must be kept independent of $Z_b^r$ according to the functional dependence graph (Fig.~\ref{fig:fdg_secrecy}) to ensure the distribution of $Z$ remains i.i.d. across blocks
\item
$\hat{M}_0'^{(b)}$ is kept independent from $Z_b^r$ to allow Encoder~1 to possess a {\em local} randomness that is separate from the common randomness shared with Encoder~2: Resolvability analysis showed us that having a local randomness at Encoder~1 can be beneficial for achievable rates.
\end{itemize}

Let $\bar{I}(\cdot ; \cdot)$ be the mutual information according to $\bar{P}$
\begin{align}
\bar{I}(M_0^{(b)}&,M_1^{(b)},M_1'^{(b)},M_1''^{(b)},\hat{M}_0^{(b)},\hat{M}_0'^{(b)},M_2^{(b)};Z_b^r)\nonumber\\
& = \mathbb{D}(\bar{P}_{M_0^{(b)} M_1^{(b)} M_1'^{(b)} M_1''^{(b)} \hat{M}_0^{(b)} \hat{M}_0'^{(b)} M_2^{(b)} Z_b^r}||\bar{P}_{M_0^{(b)} M_1^{(b)} M_1'^{(b)} M_1''^{(b)} \hat{M}_0^{(b)} \hat{M}_0'^{(b)} M_2^{(b)}} \bar{P}_{Z_b^r})\nonumber\\
& \leq \mathbb{D}(\bar{P}_{M_0^{(b)} M_1^{(b)} M_1'^{(b)} M_1''^{(b)} \hat{M}_0^{(b)} \hat{M}_0'^{(b)} M_2^{(b)} Z_b^r}||\bar{P}_{M_0^{(b)} M_1^{(b)} M_1'^{(b)} M_1''^{(b)} \hat{M}_0^{(b)} \hat{M}_0'^{(b)} M_2^{(b)}} Q_Z^{\otimes r})\label{eq:mando}
\end{align}
 
 $\mathbb{D}(\bar{P}_{M_0^{(b)} M_1^{(b)} M_1'^{(b)} M_1''^{(b)} \hat{M}_0^{(b)} \hat{M}_0'^{(b)} M_2^{(b)} Z_b^r}||\bar{P}_{M_0^{(b)} M_1^{(b)} M_1'^{(b)} M_1''^{(b)} \hat{M}_0^{(b)} \hat{M}_0'^{(b)} M_2^{(b)}} Q_Z^{\otimes r})$ can be shown, similar to Appendix~\ref{appendix:resolv_rates}, to vanish exponentially with $r$ if:
\begin{align}
 \rho_1'' &>I(U;Z), \label{eq:zoka1}\\
 \rho_1'+\rho_1'' &> I(U,X_1;Z),\\
 \rho_1'+\rho_1''+\rho_2 &> I(X_1,X_2;Z),\\
 \rho_1''+\rho_2 &>I(U,X_2;Z) \label{eq:zoka2}. 
\end{align}

Define $M^{(a:b)}=\{ M^{(a)}, \dots, M^{(b)} \}$ and $Z^{(1:b),r}=\{ Z_1^r, \dots, Z_b^r \}$.
\begin{align}
 & \bar{I}(M_1^{(1:b)},M_2^{(1:b)};Z^{(1:b),r})\nonumber\\
 &\leq \bar{I}(M_0^{(b)},M_1^{(1:b)},M_1'^{(b)},M_1''^{(b)},\hat{M}_0^{(b)},\hat{M}_0'^{(b)},M_2^{(1:b)};Z^{(1:b),r})\label{mika1}\\
  &= \bar{I}(M_0^{(b)},M_1^{(1:b)},M_1'^{(b)},M_1''^{(b)},\hat{M}_0^{(b)},\hat{M}_0'^{(b)},M_2^{(1:b)};Z_b^r)\nonumber\\
  &\hspaceonetwocol{0.2in}{0.2in} + \bar{I}(M_0^{(b)},M_1^{(1:b)},M_1'^{(b)},M_1''^{(b)},\hat{M}_0^{(b)},\hat{M}_0'^{(b)},M_2^{(1:b)};Z^{(1:b-1),r}|Z_b^r)\\
 &= \bar{I}(M_0^{(b)},M_1^{(b)},M_1'^{(b)},M_1''^{(b)},\hat{M}_0^{(b)},\hat{M}_0'^{(b)},M_2^{(b)};Z_b^r)\nonumber\\
 &\hspaceonetwocol{0.2in}{0.2in} + \bar{I}(M_1^{(1:b-1)},M_2^{(1:b-1)};Z_b^r|M_0^{(b)},M_1^{(b)},M_1'^{(b)},M_1''^{(b)},\hat{M}_0^{(b)},\hat{M}_0'^{(b)},M_2^{(b)})\nonumber\\
 &\hspaceonetwocol{0.2in}{0.2in} +
 \bar{I}(M_0^{(b)},M_1^{(1:b)},M_1'^{(b)},M_1''^{(b)},\hat{M}_0^{(b)},\hat{M}_0'^{(b)},M_2^{(1:b)};Z^{(1:b-1),r}|Z_b^r)\\
 &\leq \bar{I}(M_0^{(b)},M_1^{(b)},M_1'^{(b)},M_1''^{(b)},\hat{M}_0^{(b)},\hat{M}_0'^{(b)},M_2^{(b)};Z_b^r)\nonumber\\
 &\hspaceonetwocol{0.2in}{0.2in} + \bar{I}(M_1^{(1:b-1)},M_2^{(1:b-1)};M_1''^{(b-1)},Z_b^r|M_0^{(b)},M_1^{(b)},M_1'^{(b)},M_1''^{(b)},\hat{M}_0^{(b)},\hat{M}_0'^{(b)},M_2^{(b)})\nonumber\\
 &\hspaceonetwocol{0.2in}{0.2in} +
 \bar{I}(M_0^{(b)},M_1^{(1:b)},M_1'^{(b)},M_1''^{(b)},\hat{M}_0^{(b)},\hat{M}_0'^{(b)},M_2^{(1:b)};Z^{(1:b-1),r}|Z_b^r)\\
 &\stackrel{(a)}= \bar{I}(M_0^{(b)},M_1^{(b)},M_1'^{(b)},M_1''^{(b)},\hat{M}_0^{(b)},\hat{M}_0'^{(b)},M_2^{(b)};Z_b^r)\nonumber\\
 &\hspaceonetwocol{0.04in}{0.2in} +\! \bar{I}(M_1^{(1:b-1)},M_2^{(1:b-1)};Z_b^r|M_0^{(b)},M_1^{(b)},M_1'^{(b)},M_1''^{(b)},\hat{M}_0^{(b)},\hat{M}_0'^{(b)},M_2^{(b)},M_1^{(b-1)},M_1'^{(b-1)},M_1''^{(b-1)})\nonumber\\
 &\hspaceonetwocol{0.05in}{0.2in} +
 \bar{I}(M_0^{(b)},M_1^{(1:b)},M_1'^{(b)},M_1''^{(b)},\hat{M}_0^{(b)},\hat{M}_0'^{(b)},M_2^{(1:b)};Z^{(1:b-1),r}|Z_b^r)\\
 &\stackrel{(b)}{\leq}  2^{-\alpha r}+\bar{I}(M_0^{(b)},M_1^{(1:b)},M_1'^{(b)},M_1''^{(b)},\hat{M}_0^{(b)},\hat{M}_0'^{(b)},M_2^{(1:b)};Z^{(1:b-1),r}|Z_b^r)\\
 &\leq  2^{-\alpha r} \nonumber\\
  &\hspaceonetwocol{0.04in}{0.1in}  +\!\! \bar{I}(\!M_0^{(b)}\!,\!M_1^{(1:b)}\!,\!M_1'^{(b)}\!,\!M_1''^{(b)}\!,\!\hat{M}_0^{(b)}\!,\!\hat{M}_0'^{(b)}\!,\!M_2^{(1:b)}\!,Z_b^r,\!M_0^{(b-1)}\!,\!M_1'^{(b-1)}\!,\!M_1''^{(b-1)}\!,\!\hat{M}_0^{(b-1)}\!,\hat{M}_0'^{(b-1)};\!Z^{(1:b-1),r}\!) \\
 &\stackrel{(c)}{=} 2^{-\alpha r}
  +\bar{I}(M_0^{(b-1)},M_1^{(1:b-1)},M_1'^{(b-1)},M_1''^{(b-1)},\hat{M}_0^{(b-1)},\hat{M}_0'^{(b-1)},M_2^{(1:b-1)};Z^{(1:b-1),r}) \label{mika2}\\
  &\stackrel{(d)}{\leq} b \times  2^{-\alpha r} \nonumber\\ 
\end{align}
Therefore $\bar{I}(M_1,M_2;Z^n)\leq B \times 2^{-\alpha r}$ where, 
\begin{enumerate}[(a)]
    \item holds because $M_0^{(b)}=\hat{M}_0^{(b)}=M_1^{(b-1)}$, $\hat{M}_0'^{(b)}=M_1'^{(b-1)}$ and $M_1''^{(b-1)}$ is independent of $(M_1^{(1:b-1)},M_2^{(1:b-1)})$ by construction;
    \item holds because $\bar{I}(M_0^{(b)},M_1^{(b)},M_1'^{(b)},M_1''^{(b)},\hat{M}_0^{(b)},\hat{M}_0'^{(b)},M_2^{(b)};Z_b^r)\leq 2^{-\alpha r}$ by \eqref{eq:mando}-\eqref{eq:zoka2} and $M_1^{(1:b-1)},M_2^{(1:b-1)}\rightarrow M_0^{(b)},M_1^{(b)},M_1'^{(b)},M_1''^{(b)},\hat{M}_0^{(b)},\hat{M}_0'^{(b)},M_2^{(b)},M_1^{(b-1)},M_1'^{(b-1)},M_1''^{(b-1)}\rightarrow Z_b^r$ (see Fig.~\ref{fig:fdg_secrecy});
    \item holds because $M_0^{(b)},M_1^{(b)},M_1'^{(b)},M_1''^{(b)},\hat{M}_0^{(b-1)},\hat{M}_0'^{(b)},M_2^{(b)},Z_b^r\rightarrow  M_0^{(b-1)},M_1^{(1:b-1)},M_1'^{(b-1)},M_1''^{(b-1)},\hat{M}_0^{(b-1)},\hat{M}_0'^{(b-1)},M_2^{(1:b-1)}  \rightarrow Z^{(1:b-1),r}$ (see Fig.~\ref{fig:fdg_secrecy}).
    \item holds by repeating \eqref{mika1}-\eqref{mika2} $b-1$ times.
\end{enumerate}

Next we show that $I(M_1^{(1:b)},M_2^{(1:b)};Z^{(1:b),r})$ is not too different from $\bar{I}(M_1^{(1:b)},M_2^{(1:b)};Z^{(1:b),r})$. 
\begin{align}
&I(M_1^{(1:b)},M_2^{(1:b)};Z^{(1:b),r}) \nonumber\\
&\hspaceonetwocol{0.2in}{0.1in} = \mathbb{D}(P_{M_1^{(1:b)} M_2^{(1:b)} Z^{(1:b),r}} || P_{M_1^{(1:b)} M_2^{(1:b)}} P_{Z^{(1:b),r}}) \nonumber\\
&\hspaceonetwocol{0.2in}{0.1in} \stackrel{(a)}{\leq} \mathbb{D}(P_{M_1^{(1:b)} M_2^{(1:b)} Z^{(1:b),r}} || P_{M_1^{(1:b)} M_2^{(1:b)}} Q_Z^{\otimes br}) \nonumber\\
&\hspaceonetwocol{0.2in}{0.1in}=\sum_{m_1^{(1:b)},m_2^{(1:b)},z^{(1:b),r}} P(m_1^{(1:b)},m_2^{(1:b)},z^{(1:b),r}) \log \frac{P(m_1^{(1:b)},m_2^{(1:b)},z^{(1:b),r})}{\bar{P}(m_1^{(1:b)},m_2^{(1:b)},z^{(1:b),r})}\nonumber\\
&\hspaceonetwocol{0.6in}{0.15in}+ \sum_{m_1^{(1:b)},m_2^{(1:b)},z^{(1:b),r}} P(m_1^{(1:b)},m_2^{(1:b)},z^{(1:b),r}) \log \frac{\bar{P}(m_1^{(1:b)},m_2^{(1:b)},z^{(1:b),r})}{P(m_1^{(1:b)},m_2^{(1:b)})Q_Z^{\otimes br}}
\nonumber\\
&\hspaceonetwocol{0.6in}{0.15in}+ \mathbb{D}(\bar{P}_{M_1^{(1:b)} M_2^{(1:b)} Z^{(1:b,r)}}||P_{M_1^{(1:b)} M_2^{(1:b)}} Q_Z^{\otimes br})-\mathbb{D}(\bar{P}_{M_1^{(1:b)} M_2^{(1:b)} Z^{(1:b),r}}||P_{M_1^{(1:b)} M_2^{(1:b)}} Q_Z^{\otimes br})\nonumber\\
&\hspaceonetwocol{0.2in}{0.1in} = \mathbb{D}(P_{M_1^{(1:b)} M_2^{(1:b)} Z^{(1:b),r}} || \bar{P}_{M_1^{(1:b)} M_2^{(1:b)} Z^{(1:b),r}}) + \mathbb{D}(\bar{P}_{M_1^{(1:b)} M_2^{(1:b)} Z^{(1:b),r}}||P_{M_1^{(1:b)} M_2^{(1:b)}} Q_Z^{\otimes br}) \nonumber\\
&\hspaceonetwocol{0.6in}{0.15in}+ \sum_{m_1^{(1:b)},m_2^{(1:b)},z^{(1:b),r}}(P_{M_1^{(1:b)} M_2^{(1:b)} Z^{(1:b),r}}-\bar{P}_{M_1^{(1:b)} M_2^{(1:b)} Z^{(1:b),r}})\log\frac{\bar{P}(m_1^{(1:b)},m_2^{(1:b)},z^{(1:b),r})}{P(m_1^{(1:b)},m_2^{(1:b)})Q_Z^{\otimes br}}\nonumber\\
&\hspaceonetwocol{0.2in}{0.1in} \stackrel{(b)}{\leq} \mathbb{D}(P_{M_1^{(1:b)} M_2^{(1:b)} Z^{(1:b),r}} || \bar{P}_{M_1^{(1:b)} M_2^{(1:b)} Z^{(1:b),r}}) + \mathbb{D}(\bar{P}_{M_1^{(1:b)} M_2^{(1:b)} Z^{(1:b),r}}||\bar{P}_{M_1^{(1:b)} M_2^{(1:b)}} Q_Z^{\otimes br}) \nonumber\\
&\hspaceonetwocol{0.6in}{0.15in}+ \log \frac{1}{\mu} \mathbb{V}(P_{M_1^{(1:b)} M_2^{(1:b)} Z^{(1:b),r}},\bar{P}_{M_1^{(1:b)} M_2^{(1:b)} Z^{b,r}})\nonumber\\
&\hspaceonetwocol{0.2in}{0.1in} \stackrel{(c)}{\leq} 2\log \frac{1}{\mu} \mathbb{V}(P_{M_1^{(1:b)} M_2^{(1:b)} Z^{(1:b),r}},\bar{P}_{M_1^{(1:b)} M_2^{(1:b)} Z^{(1:b),r}}) + \mathbb{D}(\bar{P}_{M_1^{(1:b)} M_2^{(1:b)} Z^{b,r}}||\bar{P}_{M_1^{(1:b)} M_2^{(1:b)}} \bar{P}_{Z^{(1:b),r}})\nonumber\\
&\hspaceonetwocol{0.6in}{0.15in}+\mathbb{D}(\bar{P}_{Z^{(1:b),r}}||Q_Z^{\otimes br}) \nonumber\\
&\hspaceonetwocol{0.2in}{0.1in}= 2\log \frac{1}{\mu} \mathbb{V}(P_{M_1^{(1:b)} M_2^{(1:b)} Z^{(1:b),r}},\bar{P}_{M_1^{(1:b)} M_2^{(1:b)} Z^{(1:b),r}})+\bar{I}(M_1^{(1:b)},M_2^{(1:b)}; Z^{(1:b),r})\nonumber\\
&\hspaceonetwocol{0.6in}{0.15in}+\mathbb{D}(\bar{P}_{Z^{(1:b),r}}||Q_Z^{\otimes br})\label{nuu}
\end{align}
where
\begin{enumerate}[(a)]
    \item follows by adding $\mathbb{D}(P_{Z^{b,r}}||Q_Z^{\otimes br})$;
    \item follows because $\bar{P}_{M_1^{(1:b)}, M_2^{(1:b)}}=P_{M_1^{(1:b)}, M_2^{(1:b)}}$, $(P_{M_1^{(1:b)} M_2^{(1:b)} Z^{(1:b),r}}-\bar{P}_{M_1^{(1:b)} M_2^{(1:b)} Z^{(1:b),r}})\leq|P_{M_1^{(1:b)} M_2^{(1:b)} Z^{(1:b),r}}-\bar{P}_{M_1^{(1:b)} M_2^{(1:b)} Z^{(1:b),r}}|$ and by defining $\mu \triangleq \min_{z^{b,r}} Q_Z^{\otimes br}(z^{b,r})$;
    \item follows by Lemma~\ref{lemma1} and because $\mathbb{D}(\bar{P}_{M_1^{(1:b)} M_2^{(1:b)} Z^{(1:b),r}}||\bar{P}_{M_1^{(1:b)} M_2^{(1:b)}} Q_Z^{\otimes br})=\mathbb{D}(\bar{P}_{M_1^{(1:b)} M_2^{(1:b)} Z^{b,r}}||\bar{P}_{M_1^{(1:b)} M_2^{(1:b)}} \bar{P}_{Z^{(1:b),r}})+\mathbb{D}(\bar{P}_{Z^{(1:b),r}}||Q_Z^{\otimes br})$.
\end{enumerate}
The first and third terms of~\eqref{nuu} vanish exponentially with $br$ similar to Section~\ref{strictlycausal}. 

We now derive an achievable rate region by choosing values for $\rho_1'$, $\rho_1''$, $\rho_2$, $R_1$ and $R_2$ that satisfy the constraints for secrecy and probability of error. We find it more convenient to separately derive achievable rate regions under the two conditions $H(X_1|U) \lessgtr I(U,X_1;Y)$, and then merge them.

When $H(X_1|U) > I(U,X_1;Y)$, The following rates satisfy all error and secrecy constraints:
 \begin{align*}
     \rho_1''&=I(U;Z) +\epsilon, \\
     \rho_1'&=I(X_1;Z|U) +\epsilon, \\ 
     \rho_2&=I(X_2;Z|X_1,U)+\epsilon, \\
     R_1&=H(X_1|U)-I(U,X_1;Z)-2\epsilon,\\
     R_2&=I(X_1,X_2;Y)-I(X_2;Z|X_1,U)-H(X_1|U)-\epsilon,
 \end{align*}
 and the same is true for the following rates:
  \begin{align*}
     \rho_1''&=I(U,X_2;Z) +\epsilon, \\
     \rho_1'&=I(X_1,X_2;Z)-I(U,X_2;Z) +\epsilon, \\ 
     \rho_2&=\epsilon, \\
     R_1&=I(X_1,X_2;Y)-I(X_2;Y|X_1,U)-I(X_1,X_2;Z)-2\epsilon,\\
     R_2&=I(X_2;Y|X_1,U)-\epsilon.
 \end{align*}
 Considering the above two corner points, the following rate region is achievable.
 \begin{align*}
 R_1&\leq H(X_1|U)-I(U,X_1;Z)\\
 R_2&\leq I(X_2;Y|X_1,U)\\
 R_1+R_2&\leq I(X_1,X_2;Y)-I(X_1,X_2;Z)
 \end{align*}
 which is identical to Eq.~\eqref{eq:proposition3} absent one of the two sum rate constraints.
 
 
When $H(X_1|U) \leq I(U,X_1;Y)$, the following rates satisfy all error and secrecy constraints:
 \begin{align*}
     \rho_1''&=I(U;Z) +\epsilon, \\
     \rho_1'&=I(X_1;Z|U) +\epsilon, \\ 
     \rho_2&=I(X_2;Z|X_1,U)+\epsilon, \\
     R_1&=H(X_1|U)-I(U,X_1;Z)-2\epsilon,\\
     R_2&=I(X_2;Y|X_1,U)-I(X_2;Z|X_1,U)-\epsilon,
 \end{align*}
 and the same is true for the following rates:
  \begin{align*}
     \rho_1''&=I(U,X_2;Z) +\epsilon, \\
     \rho_1'&=I(X_1,X_2;Z)-I(U,X_2;Z) +\epsilon, \\ 
     \rho_2&=\epsilon, \\
     R_1&=H(X_1|U)-I(X_1,X_2;Z)-2\epsilon,\\
     R_2&=I(X_2;Y|X_1,U)-\epsilon.
 \end{align*}
Considering the above two corner points, the following rate region is achievable.
 \begin{align*}
 R_1&\leq H(X_1|U)-I(U,X_1;Z)\\
 R_2&\leq I(X_2;Y|X_1,U)\\
 R_1+R_2&\leq H(X_1|U)+I(X_2;Y|X_1,U)-I(X_1,X_2;Z)
 \end{align*}
 which is again identical to Eq.~\eqref{eq:proposition3} absent one of the two sum rate constraints.
 
 Thus far, we have two achievable rate regions for the two conditions $ H(X_1|U) \lessgtr I(U,X_1;Y)$, and the overall achievable rate region is usually specified as the union of the two. However, a more compact representation is possible via the following useful information inequality:
 \[
 H(X_1|U) \lessgtr I(U,X_1;Y) \quad \Rightarrow \quad   H(X_1|U)+ I(X_2;Y|X_1,U) \lessgtr I(X_1,X_2;Y)
 \]
 which holds because of $I(X_1,X_2;Y)=I(U,X_1,X_2;Y)$ and the chain rule. It then follows that the smaller of the two derived sum rate constraints is always active. Therefore we can simplify the expression of the achievable region by using the intersection of the two sum rate constraints.
 
 This concludes the proof of Proposition~\ref{strictly_secrecy}.

\subsection{Achievability: Strong Secrecy of \ac{MAC} with Causal Cribbing }
\label{appendix:mac_causal}

This proof is similar to the achievability proof of the strictly-causal case presented in Appendix~\ref{appendix:mac_strictlycausal} and we again use a Shannon strategy for generating $X_2$ rather than codewords \cite{cribbingmeulen}.  Using the same notation as in Section~\ref{causal} for the strategies $T$, we find from Proposition~\ref{strictly_secrecy} that rate pairs $(R_1,R_2)$ satisfying the following secrecy are achievable with strictly-causal cribbing:
\begin{align*}
R_1 &< H(X_1|U)-I(U,X_1;Z),\\
R_2 &< I(T;Y|X_1,U),\\
R_1+R_2 &< I(X_1,T;Y)-I(X_1,T;Z),
\end{align*}
for any joint distribution $P_{UX_1TYZ}\eqdef P_{U}P_{X_1|U}P_{T|U}W^{+}_{YZ|X_1T}$ with marginal $Q_Z$. Restricting the distribution to satisfy $P_{UX_1TYZ}\eqdef P_{U} P_{X_1}P_{T}W^{+}_{YZ|X_1T}$ yields:
\begin{align*}
H(X_1|U)&=H(X_1), \\
I(U,X_1;Z)&=I(X_1;Z)\\
I(T;Y|X_1,U)&=I(T,X_2;Y|X_1,U)=I(X_2;Y|X_1), \\
I(X_1,T;Y)&=I(X_1,X_2,T;Y)=I(X_1,X_2;Y),\\
I(X_1,T;Z)&=I(X_1,X_2,T;Z)=I(X_1,X_2;Z).
\end{align*}
and $P(x_1,x_2,y,z)=P(x_1)\sum_{t:t(x_1)=x_2}P(t)W(y,z|x_1,x_2)$.
 To complete the proof, we again follow \cite[Eq.~(44)]{cribbingmeulen} to note that for an arbitrary distribution $P^*(x_1,x_2)$ there exists a product distribution $P(x_1,t)=P(x_1)P(t)$ such that $P^*(x_1,x_2)=P(x_1)\sum_{t:t(x_1)=x_2}P(t)$.


\bibliographystyle{IEEEtran}
\bibliography{IEEEabrv,noha}
\end{document}